\documentclass[pra,aps,twocolumn,floatfix,longbibliography]{revtex4-1}
\usepackage[latin9]{inputenc}
\usepackage{graphicx}
\usepackage{amssymb}
\usepackage{amsmath}
\usepackage{amsthm}
\usepackage{amsfonts}
\usepackage{epsfig}
\usepackage{wrapfig}
\usepackage{refcount}
\usepackage{bm}
\usepackage{hyperref}
\usepackage{verbatim}
\usepackage{natbib}
\usepackage[margin=1in,nofoot]{geometry}

\hypersetup{
    colorlinks=true,       % false: boxed links; true: colored links
    linkcolor=cyan,          % color of internal links
    citecolor=magenta,        % color of links to bibliography
    filecolor=magenta,      % color of file links
    urlcolor=cyan,           % color of external links
    runcolor=cyan
}

\newcommand {\be}{\begin{equation}}
\newcommand {\ee}{\end{equation}}
\newcommand {\bea}{\begin{eqnarray}}
\newcommand {\eea}{\end{eqnarray}}

\newcommand{\cl}{\mathcal}

\usepackage[margin=1in,nofoot]{geometry}

\newtheorem{proposition}{Proposition}
\newtheorem{definition}{Definition}
\newtheorem{lemma}{Lemma}

\begin{document}

\title{Classifying quantum data by dissipation}

\author{Jeffrey Marshall}
\affiliation{Department of Physics and Astronomy, and Center for Quantum Information
Science \& Technology, University of Southern California, Los Angeles,
CA 90089-0484}

\author{Lorenzo Campos Venuti}

\affiliation{Department of Physics and Astronomy, and Center for Quantum Information
Science \& Technology, University of Southern California, Los Angeles,
CA 90089-0484}

\author{Paolo Zanardi}

\affiliation{Department of Physics and Astronomy, and Center for Quantum Information
Science \& Technology, University of Southern California, Los Angeles,
CA 90089-0484}

\begin{abstract}
We investigate a general class of dissipative quantum circuit capable of computing arbitrary Conjunctive Normal Form (CNF) Boolean formulas.
In particular, the clauses in a CNF formula define a local generator of Markovian quantum dynamics which acts on a network of qubits. Fixed points of this dynamical system encode the evaluation of the CNF formula.
The structure of the corresponding quantum map partitions the Hilbert space into sectors, according to decoherence-free subspaces (DFSs) associated with the dissipative dynamics.
These sectors then provide a natural and consistent way to  classify quantum data (i.e. quantum states).  
Indeed, the attractive fixed points of the network allow one to learn the sector(s) for which some particular quantum state is associated.
We show how this structure can be used to dissipatively prepare quantum states (e.g. entangled states), and outline how it may be used to generalize certain classical computational learning tasks.
\end{abstract}
\maketitle
\section{Introduction}
It has been shown that the dissipative model of quantum computation, described by Ref.~\cite{verstraete2009quantum}, is equivalent in computational power to the gate model \cite{verstraete2009quantum,jens-martin}.
In fact, since there is some built-in robustness to errors within this framework of dissipative quantum circuits \cite{verstraete2009quantum,dissi_2nd_order:Zanardi2016,DGM}, it provides a potential route to large-scale universal quantum computation. 
Moreover, since quantum states can be constructed to be  attractive fixed points of such a dissipative dynamical system, state preparation is a very natural application of these circuits \cite{kraus-prep}.
Another key area where quantum technology may prove to be beneficial is the field of machine learning \cite{Bernstein-Vazirani-learningTheory,HHL,top-data-analysis,q-hopfield}, several techniques of which actively exploit dissipative quantum dynamics \cite{monras-hmm,schuld-q-walk,quest_qnn,open-hopfield,q-hopfield2}.

In this paper we add to the growing list of dissipation assisted techniques for use in quantum computing. As such, this provides yet another example where engineered dissipation can be understood as a resource. In particular we describe a  dissipative quantum network capable of computing Conjunctive Normal Form (CNF) [and also Disjunctive Normal Form (DNF)] formulas, which connects all three of the aforementioned topics of interest: computation, state preparation, machine learning. 
Formulas of this type are of particular interest in the context of computational learning theory \cite{kearns-vazirani}.

We achieve this by introducing a local Lindbladian which defines the dynamics of a quantum system. The associated fixed points of this dynamics encode the evaluation of the clauses in the CNF formula. Recently similar techniques have been used to efficiently  transfer quantum states \cite{dissi-state-transfer}.

The circuit construction we introduce preserves quantum coherence between `classical' states (computational basis states) which have the same evaluation on the clauses in the CNF formula.
As such, this provides a  natural way to categorize quantum data, i.e., according to the resulting subspace partitioning of the Hilbert space, upon evolving the dissipative network. These partitions are in fact decoherence-free subspaces (DFSs) of the underlying dissipative dynamics \cite{Zanardi:97c,Lidar:1998fk}.
As will be demonstrated in Sect.~\ref{sect:dissi-classifier}, this level of classification can be particularly useful from the perspective of learning theory.

The inspiration for this work comes partly from the classical Hopfield recurrent neural network.
In this model, data is classified according to a set of `memories' or `patterns', to which the network dynamically evolves towards, i.e. attractive fixed points of the dynamics.
In this manner, the space of classical bit strings is partitioned according to the set of memory states.
In our model, which is similar in spirit, data is classified according to the fixed points of a dissipative quantum evolution, and the full space is partitioned into sectors according to DFSs.
Our model however in fact goes further in the sense that it can in addition classify quantum data, i.e. superpositions of `classical' states.

Based on the model we introduce, we generalize certain classical notions of data classification, providing a consistent framework from which to classify quantum data. We provide a basic outline as to how this can be used in the context of computational learning theory, where the data is quantum in nature.

Another use of this partitioning of the Hilbert space is in the  preparation of quantum states -- such as entangled states -- by dissipation alone. This technique is somewhat different to previous dissipative state preparation techniques, such as described in Refs.~\cite{kraus-prep,DGM}, which effectively use dissipation to enact the desired unitary. Our method relies on finding a CNF formula which will split the space into DFSs containing the desired states. In fact, one could prepare an ensemble of quantum states in this manner. 
We also demonstrate how to prepare quantum Probably Approximately Correct (PAC) states \cite{Valiant:PAC,qpac-bshouty-jackson}.

We discuss and bound errors within this network, which shows it actually  benefits from the engineered dissipation being as strong as possible, i.e., we are working in the strongly dissipative regime.

After briefly introducing some well known results, we will discuss our general framework, and two propositions from which all of the applications discussed above are essentially special cases. Indeed, the possibility for additional techniques arising from this methodology is promising.

\subsection{Background and Notation \label{background-sect}}

Consider the set of all bit strings of length $n$, $\cl{X} := \{0,1\}^n$.
We denote the binary variables associated with $x\in \cl{X}$ by $x_i \in \{0,1\}$, where $x = (x_1,\dots , x_n)$.

A CNF formula $C$  is a conjunction ($\land$) of clauses, where each of the clauses are disjunctions ($\lor$) of literals, $l$.
Each literal $l$ is either a variable $x_j$ or its negation $\neg x_j$, for some $1\le j \le n$.
For example, $C = C_1  \land \dots \land C_N$, with $C_i = l^{(i)}_1  \lor \dots \lor l_{k_i}^{(i)}$, where $k_i \in [n]$. We will throughout use the notation $[n]:=\{1,2,\dots , n\}$.

We state some standard results pertaining to CNF formulas and Boolean functions (see e.g. Ref.~\cite{kearns-vazirani}):
\begin{itemize}
\item[i)] Any Boolean function, $f:\cl{X} \rightarrow \{0,1\}$ can be represented by a CNF formula (although the number of clauses may be exponential in $n$).
\item[ii)] Any function $f:\cl{X} \rightarrow \{0,1\}^m$ can be represented using $m$ CNF formulas; one formula for each bit of $f(x)$.
\item[iii)] A clause $C_i$  of a CNF formula containing $k_i$ literals can be represented, using $O(k_i) = O(n)$ additional variables, by an equivalent, equisatisfiable CNF formula where each of the clauses contain at most 3 literals. Such a formula is also known as a `3-CNF'.
\end{itemize}

Throughout we will make use of the above three points, and as such will focus much of our attention on computing Boolean functions ($m=1$) with a 3-CNF representation, $C$. We will refer to the explicit evaluation of $C$ on some particular $x\in \cl{X}$, using the notation $C(x)\in\{0,1\}$.

From the perspective of machine learning, one may be directly interested in the clauses themselves since these can be thought of as the relevant (abstract) `features' being used by, for example, a machine learning algorithm to classify data. 
Imagine the task of learning a Boolean function $f$ according to some particular algorithm. 
Each time the algorithm makes an update, it can be thought of as changing the clauses in a CNF formula relating to the estimation of $f$. 
This induces a partitioning over the space $\mathcal{X}$, according to the following definition.

\begin{definition} Let $C = \land_i C_i$ be a 3-CNF formula over $\mathcal{X}$. Then,
$x,y \in \mathcal{X}$, belong to the same `partition', or `sector', of $\cl{X}$ iff $C_i(x) = C_i(y),\, \forall\, i$. 

\label{def1}
\end{definition}

For a CNF formula with $N$ clauses, this defines a partitioning of the space $\cl{X}$ into at most $2^N$ sectors, i.e. according to every possible set of features (clause evaluations) in the data.
Note that $x,y$ belonging to the same sector of course implies $C(x)=C(y)$, but the converse will not necessarily be true.

We will upgrade these general notions to define partitions over the Hilbert space, where it will be shown that Def.~\ref{def1} provides a natural starting point for defining quantum data classification by dissipation.
Practically speaking, one is often interested only in the evaluation of $C$ itself, regardless of the individual clauses.
We will in a similar manner extend the definition of the Boolean formula $C$ over classical bit strings, and give it a meaning over the Hilbert space, allowing us to write such quantities as $C(\psi)$ where $|\psi \rangle$ is a quantum state.

Since the clauses are in this context the most fundamental objects, we will first look at the `quantum evaluation' of individual clauses, $C_i(\psi)$, and then generalize this to computing 3-CNF formulas over the Hilbert space, using dissipation.

\section{Theory \label{theory-sect}}
We map the above into the quantum realm, where instead of using $n$ binary variables to encode data, one considers a system of $n$ qubits. The Hilbert space $\mathcal{H}_{\mathcal{X}} = \mathrm{span}\{|x\rangle \}_{x\in \mathcal{X}}\cong \mathbb{C}_2^{\otimes n} $ is therefore of dimension $2^n$. Here $|x\rangle := \otimes_{i=1}^n |x_i\rangle$, with $|0\rangle,|1\rangle$  eigen-states of the Pauli $z$ operator $\sigma^z:=|1\rangle \langle 1| - |0\rangle \langle 0|$.
We represent an arbitrary pure state in $\mathcal{H}_{\mathcal{X}}$ as $|\psi\rangle = \sum_{x \in \mathcal{X}} a_x |x\rangle$, for complex coefficients (amplitudes) $a_x$ which are normalized $\sum_{x \in \mathcal{X}} |a_x|^2=1$.
Occasionally we will refer explicitly to density operators $\rho \in \mathrm{L}(\mathcal{H}_{\mathcal{X}})$ which are Hermitian, positive-semidefinite, trace one, linear operators acting over the Hilbert space.

\subsection{Dissipative evaluation of 3-CNF clause \label{3sat_clause}}
Central to the dissipative network which is  to be described, is the ability to dissipatively evaluate a clause of 3-variables, i.e., $C = l_1 \lor l_2 \lor l_3$, where the $l_i$ are either $x_k$ or $\neg x_k$, for some $k \in [n]$. That is, given some $|x\rangle \in \mathcal{H}_{\mathcal{X}}$, with $x\in \cl{X}$, we will compute $C(x)$.

Let us assume the only three bits of $x$ involved in $C$ are the ${i_1},{i_2},{i_3}$-th bits, where $i_j \in [n]$. Without loss of generality we will assume these three bits are distinct.
In the quantum setting, this corresponds to three qubits.
One can construct a small dissipative network, coupling these three qubits to an ancilla qubit, $a$, via a 4-local Lindbladian \cite{Lindblad:76,Gorini:76}, as in Fig.~\ref{3cnf-clause-fig}. 

We define this Linbladian 
\begin{equation}
    \mathcal{L}\boldsymbol{\cdot} = \gamma \left( L \boldsymbol{\cdot} L^\dagger - \frac{1}{2}\{L^\dagger L,\boldsymbol{\cdot} \} \right)
    \label{eq:lindblad}
\end{equation}
by a single Lindblad (jump) operator which acts on these 4 qubits as
\begin{equation}
\begin{split}
& L = \, \Pi^\neg \otimes \sigma^-, \\
& \Pi^\neg:=|\neg l_1, \neg l_2, \neg l_3\rangle \langle \neg l_1 , \neg l_2, \neg l_3|.
\end{split}
\label{or-jump}
\end{equation}
The second term in the tensor product of $L$ acts on the ancilla qubit $a$ \footnote{{$\sigma^- = |0\rangle \langle 1|$ is defined in the eigen-basis of $\sigma^z = |1\rangle \langle 1 | - |0\rangle \langle 0|$.}}, and $\gamma >0$ is the strength of the dissipation.
The initial state of the ancilla qubit is assumed to be $|a\rangle = |1\rangle$.
The notation should be interpreted as follows. If for example, $C = x_i \lor \neg x_j \lor x_k$, the projector will be given by $\Pi^\neg = |0,1,0\rangle \langle 0,1,0|$, acting on qubits $i,j,k$.

The projector $\Pi^\neg$ projects out any state not of the form $|l_1=0,l_2=0,l_3=0\rangle$, i.e., the case in which  the clause $C$ evaluates to 0. 
In this case, the ancilla state is accordingly flipped to $|0\rangle$ by $\sigma^-$.

\begin{figure}
\vspace{0.5cm}
\includegraphics[scale=0.4]{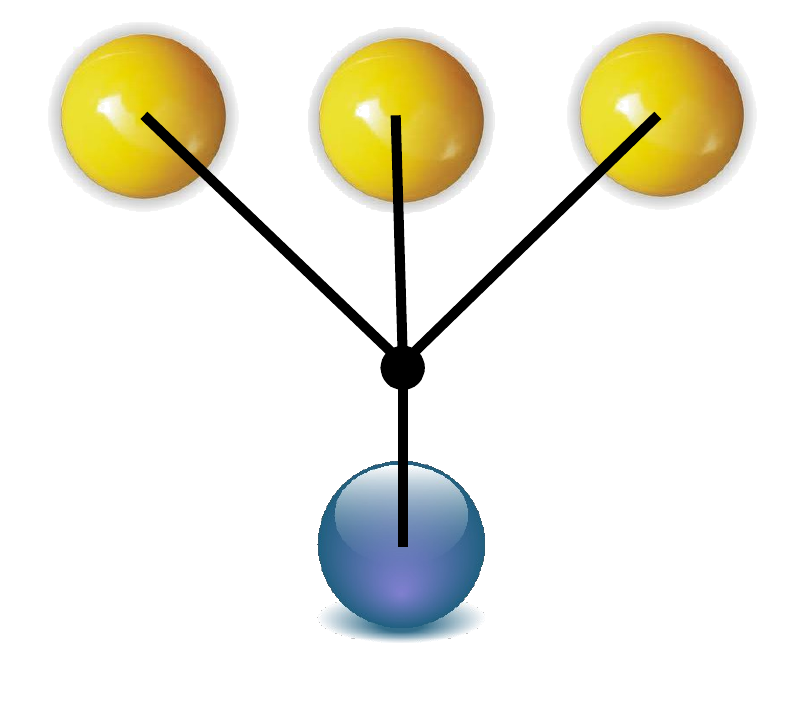}
\caption{Three input qubits (yellow, top), coupled dissipatively to an ancilla qubit (blue, bottom), via a 4-local Lindbladian (represented by the four solid lines, and the solid dot in the center).}
\label{3cnf-clause-fig}
\end{figure}

We make the following claim.

\begin{proposition} Let $\mathcal{E}_t := e^{t \mathcal{L}}$ be the evolution operator, under Lindbladian $\cl{L}$ defined by single Lindblad operator Eq.~(\ref{or-jump}). Then the evolution of a matrix element $|x\rangle \langle y| \otimes |1\rangle \langle 1|,$ where $x,y \in \mathcal{X}$, and the second term in the tensor product refers to the ancilla qubit $a$, is
\begin{widetext}
\begin{equation}
\begin{split}
\mathcal{E}_t \left( |x\rangle \langle y| \otimes |1\rangle \langle 1|\right) =
 |x\rangle \langle y | \otimes
\left\{
\begin{array}{ll}
 e^{-\gamma t}|1\rangle \langle 1| + (1-e^{-\gamma t})|c\rangle \langle c| &\,\, \text{if } C(x) = c=C(y) \\
 e^{-\gamma t/2}|1\rangle \langle 1| &\,\, \text{if } C(x) \neq C(y).
\end{array}\right.
\end{split}
\label{eq:or-evo}
\end{equation}
\end{widetext}

Note, if $C(x) = 1 = C(y)$, there is strictly no evolution, since in this case, by construction, $L|x,1\rangle = L|y,1\rangle = 0$.

\label{prop1}
\end{proposition}

\begin{proof}
Consider a disjunction of three literals, $C = l_1 \lor l_2 \lor l_3$, where $l_j$ is either $x_{i_j}$, or $\neg x_{i_j}$, for some $i_j \in [n]$. That is, $C(x)$, for $x\in\cl{X}$, is fully determined by the (distinct) bits $i_j$ in $x$, for $j=1,2,3$.

Let $\mathcal{L}$ be a Lindbladian defined by Eqs.~(\ref{eq:lindblad}), (\ref{or-jump}) with the Lindblad operator $L$ acting on (distinct) qubits $i_1,i_2,i_3$, and an ancilla qubit.

Let $x,y \in \cl{X}$. Then

\begin{equation}
\begin{split}
& \cl{L}|x\rangle \langle y| \otimes |1\rangle \langle 1| = \\
& |x\rangle \langle y| \bigotimes 
\left\{ \begin{array}{ll}
  \gamma(|0\rangle \langle 0| - |1\rangle \langle 1|) & \text{if } C(x) = 0 = C(y)\\
-\frac{\gamma}{2}|1\rangle\langle 1| & \text{if } C(x) \neq C(y) \\
0 & \text{if }  C(x) = 1 = C(y)
\end{array} \right.
\end{split}
\end{equation}

The conditions on the RHS of the above all come directly from the form of $L$:
\begin{itemize}
\item[i)] If $C(x) = 0 = C(y)$, then $x_{i_j} = \neg l_j = y_{i_j}$, for $j=1,2,3$. Hence, $L|x,1\rangle =  |x,0\rangle$, and $L|y,1\rangle = |y,0\rangle$. Similarly, $L^\dagger L |x,1\rangle = |x,1\rangle$, and $L^\dagger L |y,1\rangle = |y,1\rangle$.
\item[ii)] If  $C(x) \neq C(y)$, either $L|x,1\rangle = 0$ and $L^\dagger L|y,1\rangle = |y,1\rangle$, \textit{or} $L^\dagger L|x,1\rangle = |x,1\rangle$ and $L|y,1\rangle=0 $.
\item[iii)] If $C(x) = 1 = C(y)$, then $\exists j$ such that $x_{i_j} = l_j$, and similarly for $y$. Hence $L|x,1\rangle = 0 = L |y,1\rangle$.
\end{itemize}

Moreover, since
\begin{equation}
\cl{L}|x\rangle \langle y| \otimes |0\rangle \langle 0| = 0,
\label{eq:L0}
\end{equation}
the form of Eq.~(\ref{eq:or-evo}) is clear.
\end{proof}

We provide some additional comments regarding this result:
\begin{itemize}
\item[1)] For input state $|x\rangle \langle x|$ (with $x\in \cl{X}$), in order to guarantee with probability of $1-\epsilon$ of obtaining result $C(x)$ upon measuring the ancilla qubit, it is sufficient to evolve the network for time $t > \frac{1}{\gamma}\log \frac{1}{\epsilon}$.

That is, in order to guarantee
\begin{equation}
\langle C(x) | \mathrm{Tr}_\cl{X}\cl{E}_t\left( \rho_x \right) |C(x)\rangle \ge 1-\epsilon,
\end{equation} 
one requires $t > \frac{1}{\gamma}\log \frac{1}{\epsilon}$,
where $\rho_x:=|x\rangle \langle x| \otimes |1\rangle \langle 1|$, and $\mathrm{Tr}_\cl{X}$ is the partial trace, tracing out $\cl{H}_\cl{X}$.

 \item[2)] Coherences $|x\rangle \langle y|$ ($x\neq y$) are preserved, unless $C(x) \neq C(y)$. 
 
\item[3)] If one wishes to compute a clause of just one or two literals (instead of three), the construction and general result is exactly the same except the projector $\Pi^\neg$ of the Lindblad operator (\ref{or-jump}) acts on the appropriate one or two qubits respectively.

\item[4)] One can compute a disjunctive normal form (DNF) clause $l_1 \land l_2 \land l_3$ in a similar manner by replacing $\Pi^\neg$ by $\Pi = |l_1,l_2,l_3\rangle \langle l_1,l_2,l_3|$ and $\sigma^-$ by $\sigma^+$ in Eq.~(\ref{or-jump}). In this case the ancilla qubit should be initialized as $|0\rangle$. As such, all of our subsequent results can be equivalently phrased through DNF formulas.

 \end{itemize}
 
 Point 2) in particular implies the existence of decoherence-free subspaces \cite{Zanardi:97c,Lidar:1998fk}, defined according to
 \begin{equation}
 \mathrm{DFS}_c := \mathrm{span}\left\{ |x \rangle\,:\, C(x)=c  \right\},
 \label{eq:DFS1}
 \end{equation}
 for $c=0,1$.
 That is, for $|\psi\rangle \in \mathrm{DFS}_c$, one has $\mathrm{Tr}_a \left[\mathcal{E}_t (|\psi \rangle \langle \psi | \otimes |1\rangle \langle 1|)\right] = |\psi \rangle \langle \psi |$, where the partial trace $\mathrm{Tr}_a$ traces out the ancilla qubit $a$.

The Hilbert space is partitioned in accordance with this observation:
 given an arbitrary quantum input $\rho \in \mathrm{L}(\mathcal{H}_{\mathcal{X}})$ to the dissipative network, the quantum map $\mathcal{E}_t$ will evolve to the fixed point 
\begin{equation}
\lim_{t\rightarrow \infty}\mathcal{E}_t\left(\rho \otimes |1\rangle\langle 1| \right) = \sum_{c\in\{0,1\}} \Pi_c \rho \Pi_c \otimes |c\rangle \langle c|
\end{equation}
where $\Pi_c:=\sum_{x:C(x)=c}|x\rangle \langle x|$ is a projector onto $\mathrm{DFS}_c$ (note $\sum_{c\in \{0,1\}} \Pi_c = \mathbb{I}$).

That is, the Hilbert space $\cl{H}_\cl{X}$ contains two `coherent sectors' (DFSs) defined by $C$. States belonging to these sectors are preserved under map $\cl{E}_t$, and coherences between the sectors decay away exponentially as $e^{-\gamma t/2}$.

This simple dissipative network and the resulting decoherence-free subspaces can therefore be used to not only evaluate the clause $C(x)$ for $x\in \cl{X}$.
One can now give meaning to and compute quantities such as $C(\psi)$ where $|\psi \rangle \in \mathrm{DFS}_c$. In particular in this case the network will output $C(\psi)=c$.
We will provide a more rigorous definition below in Sect.~\ref{sect:full-network}.

We lastly comment, that in the infinite time limit, upon tracing out the ancilla qubit, the state of the system $\cl{H}_\cl{X}$, is the same as under a purely dephasing Lindbladian defined by jump operator $O=\sum_x C(x)|x\rangle \langle x |$. 
In particular, our scheme (in the long-time limit) is equivalent to  measurement of the observable $O$. That is, the engineered dissipation can be interpreted as  \textit{performing} the required measurement on the system.

\subsection{Dissipative evaluation of arbitrary 3-CNF \label{sect:full-network}}

Consider now the more general case where one wishes to evaluate a 3-CNF formula consisting of $N$ clauses $C = \land_{i=1}^N C_i$. We show how to construct a dissipative network to achieve this, which uses $N$ ancilla qubits, each of which is coupled to at most 3 of the input qubits, and evolves in a similar manner as Eq.~(\ref{eq:or-evo}). 
The full Hilbert space, of the `input' qubits, and  ancilla qubits is therefore of the form $\cl{H} = \cl{H}_\cl{X} \otimes \cl{H}_a$, where $\cl{H}_a$ is of dimension $2^N$.
We provide an illustration of this set-up in Fig.~\ref{3cnf-fig}.

The structure is now much richer, and we will see the Hilbert space is actually partitioned into (at most) $2^N$ DFSs, according to the output of each of the $N$ clauses.

The full Lindbladian acting on the system of $n+N$ qubits is now given by $\mathcal{L} = \sum_{i=1}^N\mathcal{L}_i$, where $\cl{L}_i$ is a Lindblad generator defined by a single jump operator, of the form Eq.~(\ref{or-jump}), which acts on the qubits defined by the $i$-th clause in $C$.

\begin{proposition} For $\mathcal{E}_t:=e^{t\mathcal{L}}$, with $\mathcal{L}$ defined as above, the infinite time evolution of matrix element $|x\rangle \langle y | \otimes |1\rangle \langle 1|$, with $x,y\in\mathcal{X}$, is given by

\begin{equation}
\begin{split}
&\lim_{t \rightarrow \infty} \mathcal{E}_t \left( |x\rangle \langle y| \otimes |1\rangle \langle 1|^{\otimes N}\right) = \\
& 
\left\{
\begin{array}{ll}
|x\rangle \langle y | \bigotimes_{i=1}^N |c_i \rangle \langle c_i |  & \text{if } \forall i\, C_i(x) = c_i = C_i(y) \\
 0 & \text{if } \exists\, i\,\, \text{s.t. } C_i(x) \neq C_i(y).
\end{array}\right.
\end{split}
\label{full-evo}
\end{equation}

Moreover, under a finite time evolution of input state $|x\rangle$ ($x\in \mathcal{X}$), upon measuring the $N$ ancilla qubits, to successfully obtain $C_i(x),\forall i$ with probability $1-\epsilon$, it is sufficient to evolve the network for time $t \sim O(\frac{1}{\gamma}\log \frac{N}{\epsilon})$.
\label{prop2}
\end{proposition}

\begin{proof}
Given a 3-CNF, $C = \land_{i=1}^N C_i$, we construct a 4-local Lindbladian $\mathcal{L} = \sum_{i=1}^N \mathcal{L}_i$, where $\mathcal{L}_i$ is defined for each clause, with single jump operator of the form Eq.~(\ref{or-jump}). Here, $\mathcal{L}_i$ acts on the three qubits defined by clause $C_i$, and the $i$th ancilla qubit.

The full evolution is in fact easy to calculate exactly, since the individual Lindbladians, $\mathcal{L}_i$ all commute with each other. Moreover, as can be seen from Eq.~(\ref{eq:or-evo}), the action of a single $\mathcal{L}_i$ only changes the state of the $i$th ancilla qubit.

These two facts and Prop.~\ref{prop1} imply
\begin{widetext}
\begin{equation}
\begin{split}
& \mathcal{E}_t\left( |x\rangle \langle y| \otimes |1\rangle \langle 1|^{\otimes N}\right) = \prod_{i=1}^N \mathcal{E}_t^{(i)} \left( |x\rangle \langle y| \otimes |1\rangle \langle 1|^{\otimes N}\right) \\
& = |x\rangle \langle y| \bigotimes_{i=1}^N
\left\{
\begin{array}{ll}
 e^{-\gamma t}|1\rangle \langle 1| + (1-e^{-\gamma t})|c_i\rangle \langle c_i| &\,\,\, \text{if } C_i(x) = c_i=C_i(y) \\
 e^{-\gamma t/2}|1\rangle \langle 1| &\,\,\, \text{if } C_i(x) \neq C_i(y)
\end{array}\right.
\end{split}
\label{eq:proof2}
\end{equation}
\end{widetext}

where $\mathcal{E}_t := e^{t \sum_{i=1}^N \mathcal{L}_i}$, and $\mathcal{E}_t^{(i)} : = e^{t \mathcal{L}_i}$.

Taking $t\rightarrow \infty$ gives  precisely Eq.~(\ref{full-evo}).

We comment on errors associated with finite-time evolution, $t < \infty$. 
In particular, for input of the form $|x\rangle \langle x|$ ($x\in \mathcal{X}$),  the probability of correctly obtaining via a projective measurement outcome $C_i(x)$ for the $i$-th clause is $1-e^{-\gamma t}$. This is precisely the same calculation as in point 1) following Prop.~\ref{prop1}.
Since these $N$ clauses are all independent, in order to correctly evaluate $C(x) = \land_{i=1}^N C_i(x)$ with probability $1-\epsilon$, one requires
\begin{equation}
(1-e^{-\gamma t})^N > 1-\epsilon,
\label{eq:error}
\end{equation}
and hence $t > O(\frac{1}{\gamma} \log \frac{N}{\epsilon})$ \footnote{{We wish to find $t$ such that $(1-e^{-\gamma t})^N > 1-\epsilon$. Noting that $e^{-\epsilon} > 1-\epsilon$, we can instead bound $(1-e^{-\gamma t})^N > e^{-\epsilon}$. Using the identity $-\log(1-x) < \frac{x}{1-x}$ with $x=e^{-\gamma t}$ gives $\frac{Ne^{-\gamma t}}{1-e^{-\gamma t}} < \epsilon$ which can be rearranged to give the result in the main text, $t > O(\frac{1}{\gamma}\log \frac{N}{\epsilon})$.}}.

\end{proof}

\begin{figure}
\includegraphics[scale=0.35]{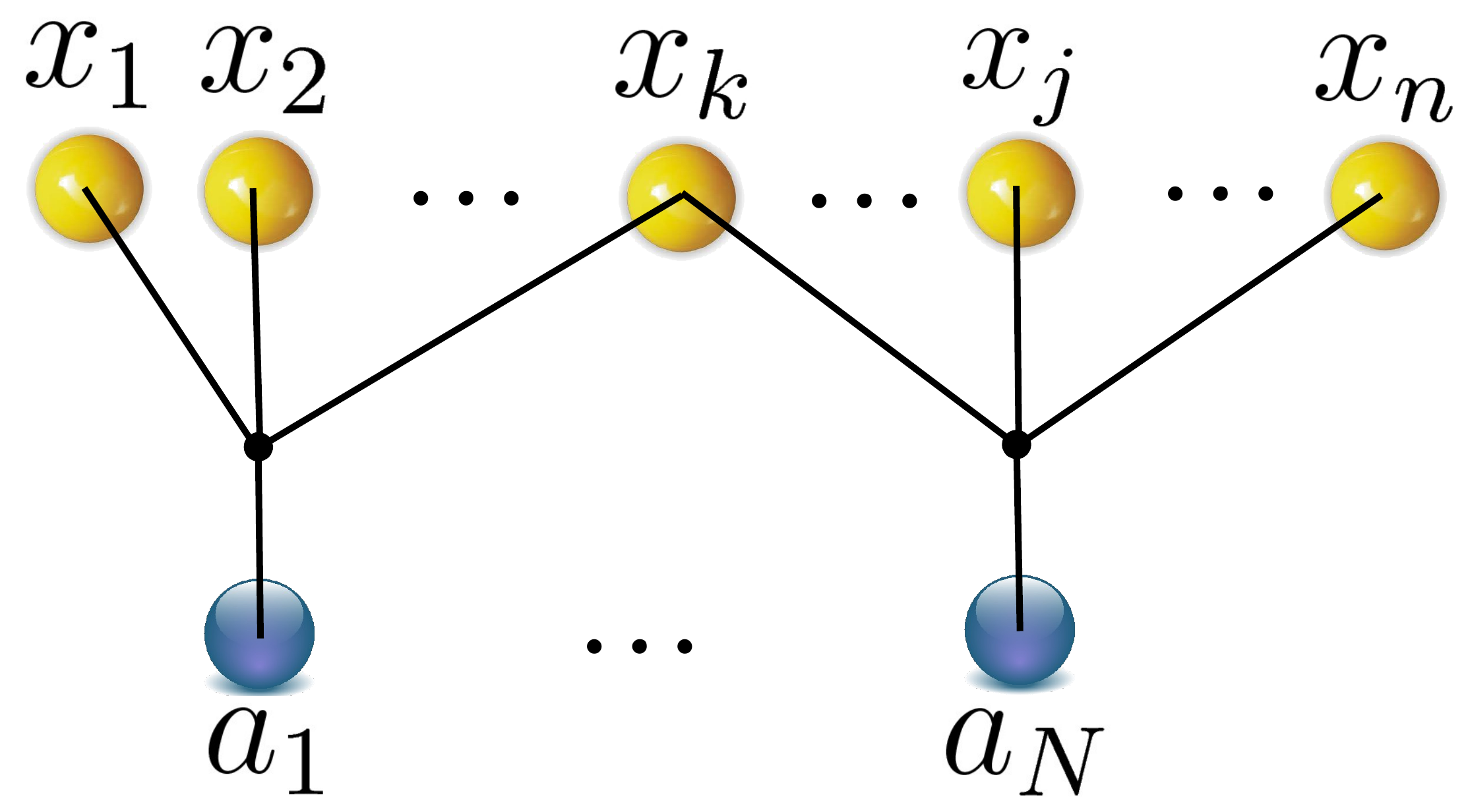}
\caption{Illustration of a dissipative computational network, capable of computing a 3-CNF with $N$ clauses, on $n$ variables. The $x_i$ represent input qubits, and the $a_i$ are ancilla qubits (where the $i$-th clause is evaluated). In this particular example, we show explicitly the clauses $C_1 = l_1 \lor l_2 \lor l_k$ and $C_N = l_k \lor l_j \lor l_n$, where $l_i$ is either $x_i$ or $\neg x_i$.}
\label{3cnf-fig}
\end{figure}

One can see from Eq.~(\ref{eq:proof2}) that this construction completely generalizes the results of the previous subsection (i.e. where $N=1$), and therefore provides a richer structure than previously described.
In particular, there are now up to $2^N$ decoherence-free subspaces, defined according to binary representation 
\begin{equation}
\bar{C}(x):=(C_1(x),\dots , C_N(x)),
\label{eq:cx}
\end{equation}
which uniquely determines the DFS to which $|x\rangle$ ($x\in\mathcal{X}$) belongs. In particular, the generalization of (\ref{eq:DFS1}) is
\begin{equation}
\mathrm{DFS}_{\bar{C}} := \mathrm{span}\left\{  |x\rangle\,:\, \bar{C}(x)=\bar{C} \right\}.
\end{equation}

If indeed $\bar{C}(x)=\bar{C}(y)$, then coherences of the form $|x\rangle \langle y|$ will be preserved under the time evolution of the dissipative quantum network, and if $|\psi\rangle \in \mathrm{DFS}_{\bar{C}}$, then
\begin{equation}
\mathrm{Tr}_{\bar{a}} \left[\mathcal{E}_t \left(|\psi \rangle \langle \psi | \otimes |1\rangle \langle 1|^{\otimes N}\right)\right] = |\psi \rangle \langle \psi |
\end{equation}
where $\mathrm{Tr}_{\bar{a}}$ traces out the $N$ ancilla qubits.

Similarly,  for an arbitrary input $\rho \in L(\mathcal{H}_{\mathcal{X}})$, the infinite time evolved state is of the form
\begin{equation}
\begin{split}
& \rho_\infty = \sum_{\bar{C}\in\{0,1\}^N} \Pi_{\bar{C}} \rho \Pi_{\bar{C}} \otimes |\bar{C}\rangle \langle \bar{C} |, \\
& \Pi_{\bar{C}}:= \sum_{x \in \mathcal{X} \, : \, \bar{C}(x) = \bar{C}} |x\rangle \langle x|,
\end{split}
\label{eq:inf-2}
\end{equation}
where $\Pi_{\bar{C}}^2 = \Pi_{\bar{C}}$  are projectors ($\sum_{\bar{C}\in\{0,1\}^N} \Pi_{\bar{C}}=\mathbb{I}$) over $\mathcal{H}_\mathcal{X}$.

In order for the final state $\rho_t:=\mathcal{E}_t (\rho \otimes |1\rangle \langle 1|^{\otimes N})$ to be $\epsilon$-close \footnote{{Here $\|\cdot \|_1$ is the trace norm: $\|X\|_1 = \mathrm{Tr}|X|$ with $|X|=\sqrt{X^\dagger X}$.}} to the steady state $\rho_{\infty}$,
\begin{equation}
    \| \rho_t - \rho_\infty \|_1 < \epsilon,
\end{equation}
requires (as shown in Appendix \ref{app:bound})
\begin{equation}
    t >  O\left(\frac{n}{\gamma} + \frac{1}{\gamma}\log \frac{N}{\epsilon}\right).
    \label{eq:inf-time}
\end{equation}

From Eq.~(\ref{eq:inf-2}) it is apparent that upon evolving an initial state $|\psi\rangle \in \mathrm{DFS}_{\bar{C}}$ (i.e. where $\Pi_{\bar{C}}|\psi\rangle = |\psi \rangle$), that one can determine $\bar{C}$
 without directly measuring or disturbing sub-system $\cl{H}_\cl{X}$ itself. That is, one can learn $\bar{C}$ passively, whilst still retaining the state $|\psi\rangle$, since here $\rho_\infty = |\psi \rangle \langle \psi | \otimes |\bar{C}\rangle \langle \bar{C}|$.
 
In this case it is apparent that one can `classify' a quantum state  $|\psi\rangle \in \mathrm{DFS}_{\bar{C}}$ according to $\bar{C}(|\psi\rangle) = \bar{C}$.
The definition is `consistent' in the sense that $\bar{C}(|x\rangle) = \bar{C}(x)$, for $x\in \mathcal{X}$, where $\bar{C}(x)$ is defined by Eq.~(\ref{eq:cx}). As such we will also write $\bar{C}(|\psi\rangle) = \bar{C}(\psi)$.

One can generalize this notion as follows:

\begin{definition} Let $|\psi\rangle \in \mathcal{H}_\mathcal{X}$. Then the `DFS-classification' of state $|\psi\rangle$ is given by a function $\tilde{C}:\mathcal{H}_\mathcal{X} \rightarrow [0,1]^N$, defined by
\begin{equation}
\tilde{C}(|\psi \rangle) \equiv \tilde{C}(\psi) : =  \sum_{\bar{C}\in\{0,1\}^N} \| \Pi_{\bar{C}} |\psi \rangle \|^2 \bar{C},
\label{def:C}
\end{equation}
where $\Pi_{\bar{C}} $ is given  in Eq.~(\ref{eq:inf-2}).
\label{def2}
\end{definition}

The interpretation of Def.~\ref{def2} should be clear: Each DFS in the Hilbert space corresponds to a unique vertex of the hypercube $[0,1]^N$. Non-vertex points correspond to states $|\psi\rangle$ which are superpositions between various DFSs, which, according to Eq.~(\ref{eq:inf-2}), are the states that become mixed  during the time evolution of the associated dissipative quantum circuit.

This motivates the straightforward observation:
\begin{lemma} Let $|\psi \rangle \in \cl{H}_{\cl{X}}$. Then,
$\tilde{C}(\psi) = \bar{C} \in \{0,1\}^N$ iff  $|\psi \rangle \in \mathrm{DFS}_{\bar{C}}$.
\label{lemma1}
\end{lemma}

We now see that Def.~\ref{def2} completely generalizes Def.~\ref{def1}; $x,y \in \cl{X}$ belong to the same partition of $\cl{X}$ according to Def.~\ref{def1}, iff $\tilde{C}(x) = \tilde{C}(y)$ in Def.~\ref{def2}.

The $i$-th clause $C_i$ evaluated on $|\psi \rangle$, is given by the $i$-th entry of $\tilde{C}(\psi)$ from Def.~\ref{def2} -- which we denote $\tilde{C}_i(\psi) \in[0,1]$ -- and shows that in the quantum case, according to this definition, Boolean logic is replaced by fuzzy logic \cite{fuzzy-logic}, in accordance with the DFS partitioning of the space. One would of course require multiple samples to estimate this quantity for a given state.

One can now also lift the definition of the Boolean function $C$ itself over classical bit strings to a function over the Hilbert space:

\begin{definition} Let $|\psi \rangle \in \cl{H}_{\cl{X}}$. The function $\hat{C}_1 : \cl{H}_{\cl{X}} \rightarrow [0,1]$ is defined by
\begin{equation}
\hat{C}_1(|\psi \rangle)  \equiv \hat{C}_1(\psi) := \| \Pi_{\bar{e}} |\psi \rangle \|^2
\end{equation}
where $\bar{e} : = (1,\dots , 1)$.
\label{def3}
\end{definition}

Def.~\ref{def3} is also fully consistent with the definition of the underlying classifier $C$, and therefore generalizes this Boolean function. In particular, for $x\in \cl{X}$, $\hat{C}_1(|x\rangle) = C(x)$. Moreover, if $|\psi \rangle$ is of the general form 
\begin{equation}
    |\psi\rangle = \sum_{x : C(x)=c} a_x |x\rangle,
\end{equation}
the network will output $\hat{C}_1(\psi) = c \in \{0,1\}$, with probability $1-\epsilon$.

Inspired by the classical evaluation of a CNF formula, $C(x) = \prod_{i=1}^N C_i(x)$ for $x\in \mathcal{X}$, one may alternatively be interested in perhaps the more natural (fuzzy) generalization:
\begin{definition} Let $|\psi \rangle \in \cl{H}_{\cl{X}}$. The function $\hat{C}_2 : \cl{H}_{\cl{X}} \rightarrow [0,1]$ is defined by
\begin{equation}
\hat{C}_2(|\psi \rangle)  \equiv \hat{C}_2(\psi) := \prod_{i=1}^N \tilde{C}_i (\psi)
\end{equation}
where $\tilde{C}_i(\psi)$ is the $i$-th entry of vector $\tilde{C}(\psi)$ in Def.~\ref{def2}.
\label{def4}
\end{definition}

This is also fully consistent with the underlying classical definition, i.e. $\hat{C}_2(|x\rangle) = C(x)$, for $x\in\mathcal{X}$. Moreover, under Def.~\ref{def4}, for $|\psi\rangle \in \mathrm{DFS}_{\bar{C}}$, $\hat{C}_2(\psi) = \prod_{i=1}^N \bar{C}_i = \hat{C}_1(\psi)$ as expected.

In this section, we have shown how one can lift classical notions pertaining to data classification in a natural manner to the quantum case. We provided a practical construction for implementing this physically, defining a local dissipative computational network, which naturally evaluates such classifiers, as in Prop.~\ref{prop2}. Not only can the network evaluate the clauses $C_i(x)$ for $x\in \cl{X}$ of a CNF formula, but also give a fuzzy logic meaning to $C_i(\psi)$ for quantum states $|\psi\rangle$, and therefore, similarly to the CNF formula $C(\psi)$ itself.
As discussed in Sect.~\ref{background-sect}, this framework applies also in the more general case where $C(x) \in \{0,1\}^m$.

We now show how to use these ideas for  practical applications relating to quantum learning theory and state preparation.

\section{Applications and Examples}

\subsection{A dissipative quantum data classifier \label{sect:dissi-classifier}}

A common classical machine learning task is, given labeled samples $(x,f(x))$ where $x\in \mathcal{X}$, and $f:\mathcal{X}\rightarrow \{0,1\}^m$, to build a classifier $C : \mathcal{X}\rightarrow \{0,1\}^m$, such that given a new sample, $y$, that $C(y) = f(y)$ with high probability.

There are many approaches which attempt to solve this problem classically, including linear regression models, Bayesian networks and artificial neural networks \cite{aima-book}.

One such quantum generalization of this type of task is  to learn $f(x)$ given samples of $n+m$ qubits of the form
\begin{equation}
|\psi_i\rangle = \sum_x a_x^{(i)} |x,f(x)\rangle ,
\label{eq:q-learn}
\end{equation}
where $i$ labels the sample.
Many well-known results in quantum learning theory are phrased in this manner, and quantum speed-ups have been demonstrated \cite{qpac-bshouty-jackson,ML-Q_world,unsupervised-speedup,Bernstein-Vazirani-learningTheory,simon:94,servedio}.
Using $|\psi_i\rangle$ as an input into our quantum network, acting only on the $n$ qubits encoding $x$ and the $N$ ancilla qubits defined by the network (i.e. leaving alone the $m$ qubits encoding $f(x)$), following from Eq.~(\ref{eq:inf-2}) the infinite time evolved state is
\begin{equation}
    \sum_{\bar{C}\in\{0,1\}^N} (\Pi_{\bar{C}} \otimes \mathbb{I}) |\psi_i\rangle \langle \psi_i | (\Pi_{\bar{C}} \otimes \mathbb{I}) \otimes |\bar{C} \rangle \langle \bar{C}|.
    \label{eq:inf3}
\end{equation}
Upon measurement of $\bar{C}$ in the $N$ ancilla qubits, the state of the $n+m$ qubits is (up to normalization)
\begin{equation}
    \sum_{x: \bar{C}(x) = \bar{C}} a_x^{(i)} |x,f(x)\rangle.
\end{equation}
That is, the sample $|\psi_i\rangle$ is projected into subspace $\mathrm{DFS}_{\bar{C}}$.

We will show by example how using the general construction outlined in the previous section, one can build a trainable quantum network,  which can be used to classify quantum data, in accordance with Defs.~\ref{def2}, \ref{def3}, \ref{def4}.

At a high level, the prescription of using our scheme to learn quantum data consists of four basic steps:
\begin{itemize}
\item[i)] Initialize network configuration: Define a system of $4$-local Lindbladians which each act on up to three input qubits, and one unique ancilla qubit, as described by Eq.~(\ref{or-jump}). 
\item[ii)] Evolve the network:
Input is a quantum state of the general form Eq.~(\ref{eq:q-learn}). Evolve network for time $O(\frac{n}{\gamma}+\frac{1}{\gamma}\log \frac{N}{\epsilon})$, where $N$ is the number of ancilla qubits, and $\epsilon$ the  error associated with finite time evolution. The resulting state will be $\epsilon$-close to Eq.~(\ref{eq:inf3}).
\item[iii)] Measurement:  Perform post-processing of the quantum data, e.g. by performing a POVM.
\item[iv)] Update: Based on the measurement result in step iii), update the network configuration and return to step ii). If no update is required and the network has converged, subsequent evaluations of the network will correctly classify new quantum data (to error $\epsilon$).
\end{itemize}

We demonstrate these general ideas by a simple and well known classical example (see e.g. Ref.~\cite{kearns-vazirani}), where a priori the space of possible functions from which one is learning is of exponential size $2^{2n}$, but where the problem class can be learned efficiently, and moreover will require  at most $2n$ ancilla qubits in our construction.

Note that the goal of the proceeding example is \textit{not} to demonstrate a quantum advantage. Rather, it is to show that using the tools outlined in this work, certain classical tasks can be generalized to a quantum setting.
Indeed, in the example below, the update procedure we use in step iv) is entirely classical in nature. 
In principle one could achieve the same outcome (learning the target function $f$) by first measuring each of the qubits in the computational basis -- destroying any quantum coherence -- and running the classical algorithm. 
By using our network however, as the algorithm progresses, the amount of quantum coherence preserved after step iii) of the procedure just mentioned, is increased.
An interesting question is whether one can gain performance advantages using an algorithm which is inherently quantum.

\subsubsection{Example: Learning the class of conjunctions \label{sect:conjunctions}}
Consider the task of learning the function $f:\mathcal{X}\rightarrow \{0,1\}$ where $f$ is guaranteed to be a pure conjunction. 
That is, $f$ is of the form $f = f_1 \land  \dots \land f_k$, where the clauses $f_i$ contain just a single literal $l_i$ that are of the form $x_j$ or $\neg x_j$ for $j \in [n]$.
Let us also write for $f$ the corresponding vector of clauses $\bar{f} = (f_1,\dots, f_k)$.
Then $\bar{f}(x) = (f_1(x),\dots,f_k(x))$ describes the evaluation of the $k$ clauses on input $x\in\cl{X}$.

This is in fact an easy problem, and the classical algorithm to solve it is to start with the hypothesis $C = x_1 \land \neg x_1 \land \dots \land x_n \land \neg x_n$ (i.e. a conjunction of all possible literals), and whenever a positive labeled sample $(y,1)$ is observed, for all $y_i=1$, remove the literal $\neg x_i$ from $C$,  and similarly for all $y_i=0$, remove the literal $x_i$ from $C$.

We will demonstrate the above in a dissipative quantum network, by utilizing $2n$ ancilla qubits, and therefore $3n$ qubits in total. The initial CNF formula defining the network is as in the previous paragraph, from which one constructs the 2-local Lindbladian $\cl{L} = \sum_{i=1}^{2n}\mathcal{L}_i$, where $\cl{L}_i$ acts on input qubit $\left \lfloor{i/2}\right \rfloor$, and the $i$-th ancilla qubit. These $\cl{L}_i$ are defined in a similar manner as in Eq.~(\ref{or-jump}), instead now with just 2 qubits.

Let us first consider the case where labeled states are promised to be of the form $(|x\rangle, f(x))$, where $x\in \mathcal{X}$.
One evolves the network under $\cl{L}$, and measures the $2n$ ancilla qubits, obtaining $C_i(x)$ where $i$ labels each ancilla qubit. We will assume that the network is evolved for a sufficiently long time at each step so that the probability of incorrectly computing $C_i(x)$ is negligible.

Given a positive sample, $f(y)=1$, one can delete any $\cl{L}_i$ from the network which results in the $i$-th ancilla qubit measuring 0, since this implies $C(y)=0$. The network is then re-set and run again, possibly now with fewer ancilla qubits. Repeating this process guarantees the convergence of $C$ to $f$.
The total number of times one must update the network in this manner is upper bound by $2n$. We show an example of the network training process in Fig.~\ref{fig:conjunction}.

It is at this point interesting to note that one can perform the exact same algorithm if one receives labeled quantum states $(|\psi\rangle,f(\psi))$ where $|\psi\rangle$ is guaranteed to be a superposition of classical states all with the same evaluation on the clauses of $f$. In this case, following Defs.~\ref{def3}, \ref{def4}, $f(\psi)\in \{0,1\}$, and the input state is   of the general form
\be
 |\psi\rangle = \sum_{x : \bar{f}(x) = \bar{{f}}'}a_x |x\rangle,
 \label{eq:psi-promise}
 \ee 
 for some $\bar{{f}}'\in \{0,1\}^k$ such that $f(\psi) = {f'}_1 \land \dots \land f'_k$, and $a_x$ arbitrary normalized complex amplitudes.
 Note, these samples can also be described equivalently as in Eq.~(\ref{eq:q-learn}), where one first measures the qubit encoding $f(x)$, since in this case $f(x)\equiv f(\psi)$, for all $ a_x$ in Eq.~(\ref{eq:psi-promise}).

Evolving the network with input $|\psi\rangle$, and
measuring the ancilla qubits -- obtaining outcome $\bar{C} = (C_1,\dots ,C_N ) \in \{0,1\}^N$  -- results in a projection $|\psi_{\bar{C}}\rangle := \Pi_{\bar{C}}|\psi\rangle / \|\Pi_{\bar{C}}|\psi\rangle\| $ (as in Eq.~(\ref{eq:inf-2})), where $N$ is the number of ancilla qubits in the current network.
Since by assumption $f(\psi) = f(\psi_{\bar{C}})$,
if $C(\psi_{\bar{C}})=C_1 \land \dots \land C_N$ disagrees with $f(\psi)$, one updates the network as above, by  deleting the conflicting $\cl{L}_i$ which define $C$.

Once the network is fully trained, given an unlabeled state $|\psi\rangle$ as above, one can evaluate $f(\psi)$ without directly measuring $|\psi\rangle$ itself, i.e., one only needs to observe the ancilla qubits. 
Moreover, since by assumption an input state $|\psi\rangle \in \mathrm{DFS}_{\bar{f}'}$ for some $\bar{f}'$, and therefore $\Pi_{\bar{f}'}|\psi\rangle = |\psi\rangle$, the initial state of the input qubits is the same as the final state; that is, one can passively learn $f(\psi)$ (i.e., classify the state $|\psi\rangle$), whilst still retaining the quantum state $|\psi\rangle$.
One can therefore classify $|\psi\rangle$,  and still use $|\psi\rangle$ for subsequent computations.

This is quite a unique occurrence, since typically to obtain any information about a quantum state requires at least some destructive measurement to take place.

\begin{figure}
\includegraphics[scale=0.4]{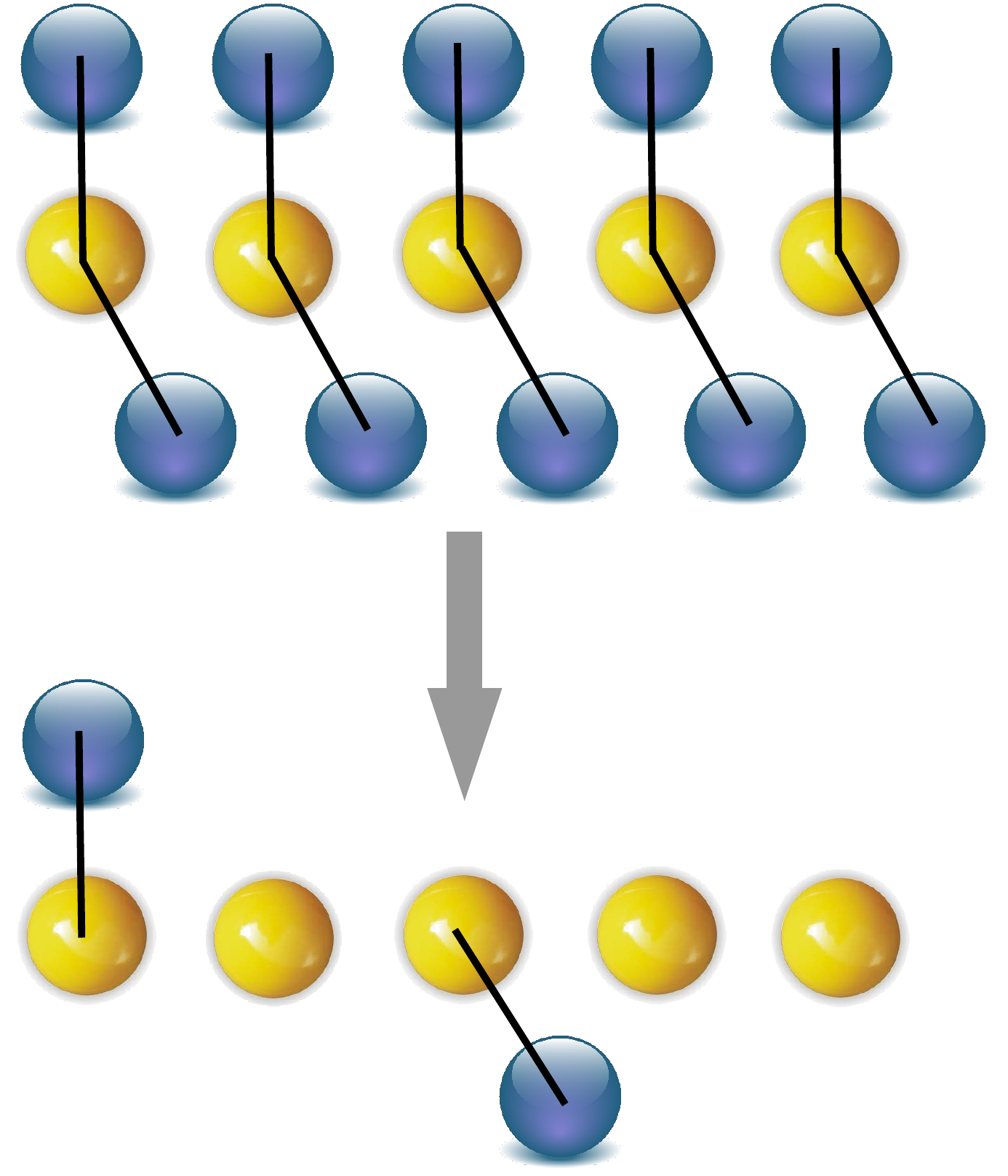}
\caption{Example of a network learning the conjunction $f = x_1 \land \neg x_3$, for $n=5$. The initial network (top) is configured to compute $x_1 \land \neg x_1 \land \dots \land x_5 \land \neg x_5$ as described by the learning algorithm of the main text. The top row of ancilla qubits (blue) compute the literals $x_i$, and the bottom row the negations $\neg x_i$. After training for a sufficiently long time so that the network has converged, the final network is given by the bottom diagram, which now has just two ancilla qubits (blue), corresponding to the computation of $x_1 \land \neg x_3$. The Hilbert space therefore has 4 associated DFSs. Input (data) qubits are shown in yellow.}
\label{fig:conjunction}
\end{figure}

\subsection{Dissipative preparation of PAC states \label{sect:PAC}}

In the quantum version \cite{qpac-bshouty-jackson} of the Probably Approximately Correct (PAC) model \cite{Valiant:PAC}, copies of the state $|\psi \rangle = \sum_{x \in \mathcal{X} } \sqrt{p(x)}|x,f(x)\rangle$ are provided, where $p(x)$ are probabilities, and $f(x)\in \{0,1\}$ a Boolean function which is to be learned.
Quantum states of this type, arising from a `quantum example oracle', are typically represented by some oracular unitary $U_f$ such that $U_f|x,b\rangle = |x,b\oplus f(x)\rangle$.
We show here that one can produce these states  in a fundamentally different manner using a dissipative (non-unitary) network of the type outlined in this work.

The goal is to prepare quantum PAC states given $\sum_{x \in \mathcal{X} } \sqrt{p(x)}|x\rangle$. The procedure which follows also applies to the more general setting where the amplitudes are complex.

We make a small modification to the Lindblad operators $L$ in Eq.~(\ref{or-jump}) as follows:
\begin{equation}
    L = \Pi^\neg \otimes \sigma^- + (\mathbb{I} - \Pi^\neg)\otimes \sigma^+.
    \label{eq:robust-lindblad}
\end{equation}

Then we have:
\begin{proposition} Let $L$ be as in Eq.~(\ref{eq:robust-lindblad}) defined by clause $C = l_1 \lor l_2 \lor l_3$. Then the infinite time evolution of matrix element $|x\rangle \langle y|$ under the network, given the ancilla state initialized as $|+\rangle = \frac{1}{\sqrt{2}}(|0\rangle + |1\rangle)$ is
\begin{equation}
\mathcal{E}_\infty (|x\rangle \langle y| \otimes |+\rangle \langle +|) = |x\rangle \langle y| \otimes |C(x)\rangle \langle C(y)|.
\end{equation}
\label{prop3}
\end{proposition}

\begin{proof}
To see this use that $\mathcal{E}_t$ is linear, and consider $\mathcal{E}_\infty( |x\rangle \langle y| \otimes |i\rangle \langle j|)$, for $i,j \in \{0,1\}$ separately. 

Consider the four cases:
\begin{equation}
    \begin{split}
        & L|x,0\rangle = 
        \left\{
\begin{array}{ll}
0  & \text{if } C(x) = 0\\
 |x,1\rangle & \text{if }  C(x)=1
\end{array}\right.
\\
        & L|x,1\rangle = 
        \left\{
\begin{array}{ll}
|x,0\rangle & \text{if } C(x) = 0\\
 0 & \text{if }  C(x)=1,
\end{array}\right.
    \end{split}
\end{equation}
and similarly $L^\dagger L |x,i\rangle = \delta_{i, \neg C(x)}|x,i\rangle$ where $i\in \{0,1\}$.

One can then write
\begin{equation}
\begin{split}
     & \cl{L}\left(|x\rangle \langle y| \otimes |i\rangle \langle j| \right) = \\ &
    \gamma |x\rangle \langle y| \otimes 
    [ \delta_{i,\neg C(x)} \delta_{j,\neg C(y)} |C(x)\rangle \langle C(y)| - 
    \\ &
      \frac{1}{2}|i\rangle \langle j| \left(\delta_{i,\neg C(x)} + \delta_{j,\neg C(y)} \right) ] .
\end{split}
\end{equation}

From this one can check -- in a similar manner as the proof of Prop.~\ref{prop1} -- that if $C(x)=C(y)$, then 
\begin{equation}
    \mathcal{E}_\infty( |x\rangle \langle y| \otimes |i\rangle \langle j|) = \delta_{i,j}|x\rangle \langle y|\otimes |C(x)\rangle \langle C(y)|.
\end{equation}

Similarly, if $C(x)\neq C(y)$, one has
\begin{equation}
    \mathcal{E}_\infty( |x\rangle \langle y| \otimes |i\rangle \langle j|) = \delta_{i,\neg j}|x\rangle \langle y|\otimes |C(x)\rangle \langle C(y)|.
\end{equation}
 That is, in this latter case, coherences $|x\rangle \langle y |$ remain only when $i\neq j$.
 
 The result follows noting that $|+\rangle \langle + | = \sum_{i} \frac{1}{2}|i\rangle \langle i| + \sum_{i\neq j}\frac{1}{2}|i\rangle \langle j |$.
 
\end{proof}

From Prop.~\ref{prop3} it immediately follows that under the network with $N$ clauses one has in the long time limit \footnote{The same arguments as in Appendix \ref{app:bound} apply to Prop.~\ref{prop3} regarding the evolution time as a function of error $\epsilon$, which is therefore still $t\sim O(\frac{n}{\gamma}+\frac{1}{\gamma}\log \frac{N}{\epsilon})$.} ($t\rightarrow  \infty$)
\begin{equation}
    \sum_{x \in \mathcal{X} } \sqrt{p(x)}|x\rangle |+\rangle^{\otimes N} \rightarrow \sum_{x \in \mathcal{X} } \sqrt{p(x)}|x\rangle |\bar{f}(x)\rangle
\end{equation}
where $|\bar{f}(x)\rangle = \otimes_{i=1}^N |f_i(x)\rangle$ with $f_i(x)$ the $i$-th clause of $f(x)$ realized as a 3-CNF.
Since $|\bar{f}(x)\rangle$ contains all the information required to compute $f(x)$, states of this type can equivalently be used in the PAC framework. The only complication is that one now has $N$ additional qubits in the scheme, instead of just 1. 
Nevertheless, if one is interested in machine learning involving CNF or DNF formulas \footnote{As mentioned in point 4 following Prop.~\ref{prop1}, our framework can equivalently be used to evaluate both CNF and DNF formulas.} this framework provides an alternative mechanism for obtaining quantum PAC states.
This is therefore relevant even in the original formulation of the quantum PAC model by  Ref.~\cite{qpac-bshouty-jackson} for which formulas in DNF over the uniform distribution can be efficiently PAC learned on a quantum computer, as compared to the best known quasi-polynomial classical algorithm.

\subsection{Probabilistic preparation of quantum states}

By Eq.~(\ref{eq:inf-2}), upon measuring state $|\bar{C}\rangle$ of the ancilla qubits, the resulting time-evolved input state is, with probability $1-\epsilon$, $\Pi_{\bar{C}} \rho \Pi_{\bar{C}}/\mathrm{Tr}(\Pi_{\bar{C}} \rho)$, where $\rho$ was the initial state.

One can use this to prepare quantum states. In particular, to prepare state $|\psi\rangle$, one requires a suitable, easy to prepare input state, $|\psi_0\rangle$, and to define an appropriate CNF formula so that, for some $\bar{C}\in \{0,1\}^N$, one has $\Pi_{\bar{C}} |\psi_0\rangle \propto |\psi\rangle$.
Then, upon measuring $\bar{C}$ in the ancilla qubits, which occurs with probability $\langle \psi_0 |\Pi_{\bar{C}}|\psi_0\rangle$, one has prepared state $|\psi\rangle$.

One could also prepare an ensemble of states in this manner.

\subsubsection{Preparation of entangled states}

Consider for simplicity, two qubits, which we wish to entangle, that  are initially prepared in the product state $|\psi_0\rangle = |+\rangle\otimes |+\rangle$.
Here, in the $z$-eigenbasis, $|+\rangle := \frac{1}{\sqrt{2}}(|0\rangle + |1\rangle)$.
We wish to prepare the Bell state $|\psi^+\rangle := \frac{1}{\sqrt{2}}(|00\rangle + |11\rangle )$.

This procedure will be effectively the same as performing a projective measurement on $|\psi_0\rangle$, although we will in fact not perform any measurement over $\mathcal{H}_{\mathcal{X}}$. This example is interesting therefore for two reasons: 1) dissipation is used directly to create entanglement, and 2) it provides a novel manner in which to actually perform a measurement of a quantum system.

To achieve this task, we take the 2-CNF, $C = (x_1 \lor \neg x_2) \land (\neg x_1 \lor x_2)=:C_1 \land C_2$, with associated projectors
\begin{equation}
\begin{split}
& \Pi_{(0,0)} = 0,\, \Pi_{(0,1)} = |01\rangle \langle 01|,\,\Pi_{(1,0)} = |10\rangle \langle 10|, \\ 
& \Pi_{(1,1)} = |00\rangle \langle 00 | + |11\rangle \langle 11|.
\end{split}
\end{equation}
From this, utilizing two ancilla qubits that are initialized in the $|1\rangle$ state, one can construct a following 3-local Lindbladian, $\cl{L} = \cl{L}_1 + \cl{L}_2$.
These Lindbladians are defined by Lindblad operators $L_1 = |01\rangle \langle 0 1|\otimes \sigma_1^-$, $L_2 = |10\rangle \langle 10|\otimes \sigma_2^-$, where the second term in the tensor product acts on a single ancilla qubit, labeled by $1,2$ respectively.
See inset of Fig.~\ref{entanglement-fig} for a schematic of the network.

This evolution will partition the space into three DFSs, and in particular, the infinite time evolved state is
\begin{equation}
\rho_\infty = \sum_{c_1,c_2 \in \{0,1\}} P_{c_1,c_2}|\psi_{c_1,c_2}\rangle \langle \psi_{c_1,c_2} | \otimes |c_1,c_2\rangle \langle c_1,c_2|,
\end{equation}
with the probability of being in each sector  after measurement of the two ancilla qubits $P_{\bar{C}} = \langle \psi_0 | \Pi_{\bar{C}} | \psi_0 \rangle$,
\begin{equation}
P_{0,0} = 0,\, P_{0,1} =  \frac{1}{4} = P_{1,0},\,P_{1,1} = \frac{1}{2}.
\end{equation}

The states in each sector are given by
\begin{equation}
\begin{split}
& |\psi_{0,1}\rangle = \frac{\Pi_{(0,1)}|\psi_0\rangle}{\sqrt{P_{0,1}}}= |01\rangle  \\
& |\psi_{1,0}\rangle \frac{\Pi_{(1,0)}|\psi_0\rangle}{\sqrt{P_{1,0}}}  =|10\rangle\\
& |\psi_{1,1}\rangle =\frac{\Pi_{(1,1)}|\psi_0\rangle}{\sqrt{P_{1,1}}}= \frac{1}{\sqrt{2}}(|00\rangle + |11\rangle).
\end{split}
\end{equation}
Note, the $(c_1,c_2)=(0,0)$ sector is empty since at least one of the two clauses in $C$ must be satisfied by construction. That is, in this case, there is no such state $|\psi_{0,0}\rangle$.

After evolving the network for a sufficiently long time, one can post-select on the ancilla qubits with success probability $\frac{1}{2}$ to pick out the Bell state:
\begin{equation}
\rho_{1,1}(t) :=\mathrm{Tr}_{\bar{a}}\left[\tilde{\Pi}_{1,1} \rho(t) \tilde{\Pi}_{1,1} \right] \xrightarrow[t \to \infty]{} |\psi^+ \rangle \langle \psi^+| 
\label{eq:rho_11}
\end{equation}
where 
\begin{equation}
\begin{split}
& \rho(t) := \cl{E}_t \left( |\psi_0\rangle \langle \psi_0| \otimes |1,1\rangle \langle 1,1| \right),\\
& \tilde{\Pi}_{1,1} : = \mathbb{I}_4 \otimes |1,1\rangle \langle 1,1|.
\end{split}
\end{equation}

In Fig.~\ref{entanglement-fig} we provide a numerical verification of this scheme, in which the state of the qubits, upon post-selecting on (1,1) on the two ancilla qubits, converges exponentially quickly to the desired maximally entangled state.

\begin{figure}
\includegraphics[scale=0.4]{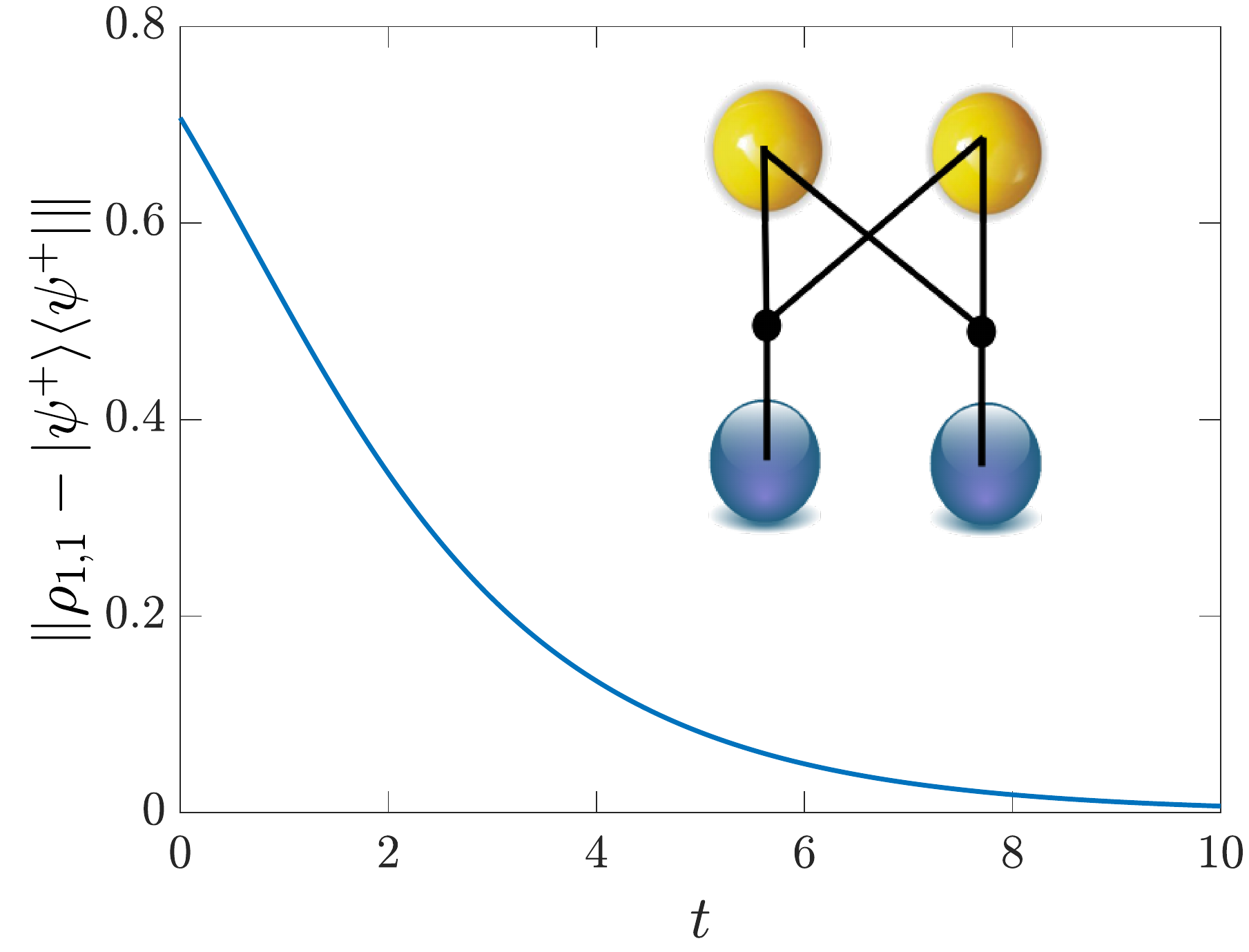}
\caption{Distance to maximally entangled Bell state $|\psi^+\rangle = \frac{1}{\sqrt{2}}(|00\rangle + |11\rangle)$, under dissipative evolution for time $t$ (in units of $1/\gamma$). Here, $\rho_{1,1}$ is the state after post-selecting on measuring 1 in both ancilla qubits Eq.~(\ref{eq:rho_11}), which occurs with probability $1/2$. Initial state is the product state $|+,+\rangle$. The norm is the maximum singular value. Inset: the connectivity of the 2+2 qubit dissipative network.}
\label{entanglement-fig}
\end{figure}

\subsubsection{Generating a superposition of all solutions to 3-SAT problem}

3-SAT problems are NP-hard optimization problems, where one must find the solution(s) $C(x)=1$ to  a 3-CNF formula $C$.

Given a 3-SAT problem $C$ over $n$ literals, one can construct a dissipative network (of 4-local Lindbladians) with $N$ ancilla qubits, where $N$ is the number of clauses in the 3-SAT problem.

The input state to the network is initialized in the maximal superposition state, $|\psi\rangle = \frac{1}{2^{n/2}} \sum_{i=1}^{2^n} |i\rangle$.
After evolving the network for time $t \sim O(\frac{n}{\gamma}+\frac{1}{\gamma}\log \frac{N}{\epsilon})$, upon measurement of $|1\rangle^{\otimes N}$ of the ancilla qubits, the resulting state of the input qubits is guaranteed (with probability $1-\epsilon$) to be
\begin{equation}
|S\rangle = \frac{1}{\sqrt{N_s}} \sum_{s=1}^{N_s} |s\rangle
\label{eq:solutions}
\end{equation}
where each classical bit string $s$ is such that $C(s) = 1$. Moreover, there are no other bit-strings with this property. That is, $|S\rangle$ is the uniform superposition of all solutions to the 3-SAT problem. We strongly stress however that in general it is exponentially unlikely to observe such a state; we are not claiming to have an efficient algorithm to solve SAT problems.

Nevertheless, this general methodology shows potential promise in obtaining a fair sampling for certain problems, which is typically challenging on current quantum optimization devices \cite{fairInUnfair}.

\section{Discussion}

We have shown how to evaluate 3-CNF (and DNF) clauses (and hence arbitrary Boolean functions) using dissipative quantum dynamics.
We did this by providing a way to construct a dissipative network consisting of 4-local Lindbladians, where upon measuring a subset of the qubits in the network, one can evaluate the CNF formula. 
We also showed that errors associated with a finite time evolution can be made arbitrarily small.
The CNF structure naturally partitions the Hilbert space into sectors in which coherence is preserved, i.e., decoherence-free subspaces. 
This provides a route to generalize classical notions of data classification, lifting the definitions to the case where the data is quantum in nature.
In particular, we achieve this task by classifying quantum data according to the partition, or DFS, to which a particular state belongs, which itself is defined by the dissipative network.

This has applications in state preparation, and perhaps more interestingly, machine learning. 
We provided one such example, but have hopefully have demonstrated the general applicability of these techniques.
Indeed, one can make a connection between the work presented here, and a Hopfield recurrent neural network, which we now comment on.

A Hopfield network is a classical pattern recognition and classification machine learning model, where the `patterns' or `memories' one wishes to associate new data, are local minima of a pre-defined energy function. 
In particular, these patterns act like fixed points of the total space $\cl{X}$ of a dynamical system (where the dynamics is described entirely by the energy function).
We  make an analogy between this classical dynamical system, and the quantum dynamical systems described in this work defined entirely by dissipative dynamics with attractive fixed points.

In our model, which is morally similar, data is classified according to the DFSs of the computational network, which are analogous to the patterns of the corresponding Hopfield network. 
However, not only can the dissipative networks presented in this work classify classical data, e.g. compute a function $C(x)\in\{0,1\}$ for  $x\in\cl{X}$, they can also classify quantum data in a consistent manner. 
Moreover, given a quantum state $|\psi\rangle$ associated to a particular memory pattern (i.e. a DFS), one can evaluate $C(\psi)$ passively, without disturbing the state $|\psi\rangle$.

We lastly mention that the same type of construction given in this work may of course be demonstrated through local unitary evolutions and single-qubit projective measurements alone. 
Though this is true, it misses perhaps the main point; a demonstration that engineered dissipation can be used as a resource for certain machine learning and classification tasks.
We have given a precise construction for implementing the tasks described above, where the engineered dissipation is essentially performing the measurements of the relevant observables, required for classifying quantum data according to Defs.~\ref{def2}, \ref{def3}, \ref{def4}.
Indeed, similar arguments can be made against \textit{any} type of dissipative based scheme.
In particular, it is known that the dissipative model of quantum computation is no more powerful than the unitary gate model, and that therefore the latter can efficiently  simulate the former \cite{jens-martin}.
Nevertheless, it is a conceptually interesting  observation that purely non-unitary dynamics can be used to carry out tasks relevant for machine learning.

This work provides a well defined and consistent starting point for generalizing certain aspects of classsical machine learning theory, of which numerous applications could benefit. We highlighted simple examples for demonstrative purposes, but hope it is clear that  the general protocols can be modified in a multitude of ways to be made applicable for different applications. It would be interesting for example to apply these techniques to learning quantum data in the PAC framework \cite{kearns-vazirani,Valiant:PAC,qpac-bshouty-jackson,qpac-optimal}.

\section{Acknowledgments}
We thank Evgeny Mozgunov for reading the paper and providing useful comments.
 P.Z. acknowledges partial support from the NSF award PHY-1819189.
 L.C.V. acknowledges partial support from the Air Force Research Laboratory award no. FA8750-18-1-004.

\bibliography{refs}

%merlin.mbs apsrev4-1.bst 2010-07-25 4.21a (PWD, AO, DPC) hacked
%Control: key (0)
%Control: author (0) dotless jnrlst
%Control: editor formatted (1) identically to author
%Control: production of article title (0) allowed
%Control: page (1) range
%Control: year (0) verbatim
%Control: production of eprint (0) enabled
\begin{thebibliography}{35}%
\makeatletter
\providecommand \@ifxundefined [1]{%
 \@ifx{#1\undefined}
}%
\providecommand \@ifnum [1]{%
 \ifnum #1\expandafter \@firstoftwo
 \else \expandafter \@secondoftwo
 \fi
}%
\providecommand \@ifx [1]{%
 \ifx #1\expandafter \@firstoftwo
 \else \expandafter \@secondoftwo
 \fi
}%
\providecommand \natexlab [1]{#1}%
\providecommand \enquote  [1]{``#1''}%
\providecommand \bibnamefont  [1]{#1}%
\providecommand \bibfnamefont [1]{#1}%
\providecommand \citenamefont [1]{#1}%
\providecommand \href@noop [0]{\@secondoftwo}%
\providecommand \href [0]{\begingroup \@sanitize@url \@href}%
\providecommand \@href[1]{\@@startlink{#1}\@@href}%
\providecommand \@@href[1]{\endgroup#1\@@endlink}%
\providecommand \@sanitize@url [0]{\catcode `\\12\catcode `\$12\catcode
  `\&12\catcode `\#12\catcode `\^12\catcode `\_12\catcode `\%12\relax}%
\providecommand \@@startlink[1]{}%
\providecommand \@@endlink[0]{}%
\providecommand \url  [0]{\begingroup\@sanitize@url \@url }%
\providecommand \@url [1]{\endgroup\@href {#1}{\urlprefix }}%
\providecommand \urlprefix  [0]{URL }%
\providecommand \Eprint [0]{\href }%
\providecommand \doibase [0]{http://dx.doi.org/}%
\providecommand \selectlanguage [0]{\@gobble}%
\providecommand \bibinfo  [0]{\@secondoftwo}%
\providecommand \bibfield  [0]{\@secondoftwo}%
\providecommand \translation [1]{[#1]}%
\providecommand \BibitemOpen [0]{}%
\providecommand \bibitemStop [0]{}%
\providecommand \bibitemNoStop [0]{.\EOS\space}%
\providecommand \EOS [0]{\spacefactor3000\relax}%
\providecommand \BibitemShut  [1]{\csname bibitem#1\endcsname}%
\let\auto@bib@innerbib\@empty
%</preamble>
\bibitem [{\citenamefont {Verstraete}\ \emph {et~al.}(2009)\citenamefont
  {Verstraete}, \citenamefont {Wolf},\ and\ \citenamefont
  {Ignacio~Cirac}}]{verstraete2009quantum}%
  \BibitemOpen
  \bibfield  {author} {\bibinfo {author} {\bibfnamefont {F.}~\bibnamefont
  {Verstraete}}, \bibinfo {author} {\bibfnamefont {M.~M.}\ \bibnamefont
  {Wolf}}, \ and\ \bibinfo {author} {\bibfnamefont {J.}~\bibnamefont
  {Ignacio~Cirac}},\ }\bibfield  {title} {\enquote {\bibinfo {title} {Quantum
  computation and quantum-state engineering driven by dissipation},}\ }\href
  {http://dx.doi.org/10.1038/nphys1342} {\bibfield  {journal} {\bibinfo
  {journal} {Nat. Phys.}\ }\textbf {\bibinfo {volume} {5}},\ \bibinfo {pages}
  {633} (\bibinfo {year} {2009})}\BibitemShut {NoStop}%
\bibitem [{\citenamefont {Kliesch}\ \emph {et~al.}(2011)\citenamefont
  {Kliesch}, \citenamefont {Barthel}, \citenamefont {Gogolin}, \citenamefont
  {Kastoryano},\ and\ \citenamefont {Eisert}}]{jens-martin}%
  \BibitemOpen
  \bibfield  {author} {\bibinfo {author} {\bibfnamefont {M.}~\bibnamefont
  {Kliesch}}, \bibinfo {author} {\bibfnamefont {T.}~\bibnamefont {Barthel}},
  \bibinfo {author} {\bibfnamefont {C.}~\bibnamefont {Gogolin}}, \bibinfo
  {author} {\bibfnamefont {M.}~\bibnamefont {Kastoryano}}, \ and\ \bibinfo
  {author} {\bibfnamefont {J.}~\bibnamefont {Eisert}},\ }\bibfield  {title}
  {\enquote {\bibinfo {title} {{Dissipative Quantum Church-Turing Theorem}},}\
  }\href {\doibase 10.1103/PhysRevLett.107.120501} {\bibfield  {journal}
  {\bibinfo  {journal} {Phys. Rev. Lett.}\ }\textbf {\bibinfo {volume} {107}},\
  \bibinfo {pages} {120501} (\bibinfo {year} {2011})}\BibitemShut {NoStop}%
\bibitem [{\citenamefont {Zanardi}\ \emph {et~al.}(2016)\citenamefont
  {Zanardi}, \citenamefont {Marshall},\ and\ \citenamefont
  {Campos~Venuti}}]{dissi_2nd_order:Zanardi2016}%
  \BibitemOpen
  \bibfield  {author} {\bibinfo {author} {\bibfnamefont {P.}~\bibnamefont
  {Zanardi}}, \bibinfo {author} {\bibfnamefont {J.}~\bibnamefont {Marshall}}, \
  and\ \bibinfo {author} {\bibfnamefont {L.}~\bibnamefont {Campos~Venuti}},\
  }\bibfield  {title} {\enquote {\bibinfo {title} {{Dissipative universal
  Lindbladian simulation}},}\ }\href {\doibase 10.1103/PhysRevA.93.022312}
  {\bibfield  {journal} {\bibinfo  {journal} {Phys. Rev. A}\ }\textbf {\bibinfo
  {volume} {93}},\ \bibinfo {pages} {022312} (\bibinfo {year}
  {2016})}\BibitemShut {NoStop}%
\bibitem [{\citenamefont {Marshall}\ \emph {et~al.}(2016)\citenamefont
  {Marshall}, \citenamefont {Campos~Venuti},\ and\ \citenamefont
  {Zanardi}}]{DGM}%
  \BibitemOpen
  \bibfield  {author} {\bibinfo {author} {\bibfnamefont {J.}~\bibnamefont
  {Marshall}}, \bibinfo {author} {\bibfnamefont {L.}~\bibnamefont
  {Campos~Venuti}}, \ and\ \bibinfo {author} {\bibfnamefont {P.}~\bibnamefont
  {Zanardi}},\ }\bibfield  {title} {\enquote {\bibinfo {title} {Modular
  quantum-information processing by dissipation},}\ }\href {\doibase
  10.1103/PhysRevA.94.052339} {\bibfield  {journal} {\bibinfo  {journal} {Phys.
  Rev. A}\ }\textbf {\bibinfo {volume} {94}},\ \bibinfo {pages} {052339}
  (\bibinfo {year} {2016})}\BibitemShut {NoStop}%
\bibitem [{\citenamefont {Kraus}\ \emph {et~al.}(2008)\citenamefont {Kraus},
  \citenamefont {B{\"u}chler}, \citenamefont {Diehl}, \citenamefont {Kantian},
  \citenamefont {Micheli},\ and\ \citenamefont {Zoller}}]{kraus-prep}%
  \BibitemOpen
  \bibfield  {author} {\bibinfo {author} {\bibfnamefont {B.}~\bibnamefont
  {Kraus}}, \bibinfo {author} {\bibfnamefont {H.~P.}\ \bibnamefont
  {B{\"u}chler}}, \bibinfo {author} {\bibfnamefont {S.}~\bibnamefont {Diehl}},
  \bibinfo {author} {\bibfnamefont {A.}~\bibnamefont {Kantian}}, \bibinfo
  {author} {\bibfnamefont {A.}~\bibnamefont {Micheli}}, \ and\ \bibinfo
  {author} {\bibfnamefont {P.}~\bibnamefont {Zoller}},\ }\bibfield  {title}
  {\enquote {\bibinfo {title} {{Preparation of entangled states by quantum
  Markov processes}},}\ }\href
  {http://link.aps.org/doi/10.1103/PhysRevA.78.042307} {\bibfield  {journal}
  {\bibinfo  {journal} {Phys. Rev. A}\ }\textbf {\bibinfo {volume} {78}},\
  \bibinfo {pages} {042307} (\bibinfo {year} {2008})}\BibitemShut {NoStop}%
\bibitem [{\citenamefont {Bernstein}\ and\ \citenamefont
  {Vazirani}(1993)}]{Bernstein-Vazirani-learningTheory}%
  \BibitemOpen
  \bibfield  {author} {\bibinfo {author} {\bibfnamefont {E.}~\bibnamefont
  {Bernstein}}\ and\ \bibinfo {author} {\bibfnamefont {U.}~\bibnamefont
  {Vazirani}},\ }\bibfield  {title} {\enquote {\bibinfo {title} {{Quantum
  Complexity Theory}},}\ }in\ \href {\doibase 10.1145/167088.167097} {\emph
  {\bibinfo {booktitle} {Proceedings of the Twenty-fifth Annual ACM Symposium
  on Theory of Computing}}},\ \bibinfo {series and number} {STOC '93}\
  (\bibinfo  {publisher} {ACM},\ \bibinfo {address} {New York, NY, USA},\
  \bibinfo {year} {1993})\ pp.\ \bibinfo {pages} {11--20}\BibitemShut {NoStop}%
\bibitem [{\citenamefont {Harrow}\ \emph {et~al.}(2009)\citenamefont {Harrow},
  \citenamefont {Hassidim},\ and\ \citenamefont {Lloyd}}]{HHL}%
  \BibitemOpen
  \bibfield  {author} {\bibinfo {author} {\bibfnamefont {A.~W.}\ \bibnamefont
  {Harrow}}, \bibinfo {author} {\bibfnamefont {A.}~\bibnamefont {Hassidim}}, \
  and\ \bibinfo {author} {\bibfnamefont {S.}~\bibnamefont {Lloyd}},\ }\bibfield
   {title} {\enquote {\bibinfo {title} {{Quantum Algorithm for Linear Systems
  of Equations}},}\ }\href {\doibase 10.1103/PhysRevLett.103.150502} {\bibfield
   {journal} {\bibinfo  {journal} {Phys. Rev. Lett.}\ }\textbf {\bibinfo
  {volume} {103}},\ \bibinfo {pages} {150502} (\bibinfo {year}
  {2009})}\BibitemShut {NoStop}%
\bibitem [{\citenamefont {Lloyd}\ \emph {et~al.}(2016)\citenamefont {Lloyd},
  \citenamefont {Garnerone},\ and\ \citenamefont
  {Zanardi}}]{top-data-analysis}%
  \BibitemOpen
  \bibfield  {author} {\bibinfo {author} {\bibfnamefont {S.}~\bibnamefont
  {Lloyd}}, \bibinfo {author} {\bibfnamefont {S.}~\bibnamefont {Garnerone}}, \
  and\ \bibinfo {author} {\bibfnamefont {P.}~\bibnamefont {Zanardi}},\
  }\bibfield  {title} {\enquote {\bibinfo {title} {{Quantum algorithms for
  topological and geometric analysis of data}},}\ }\href
  {http://dx.doi.org/10.1038/ncomms10138} {\bibfield  {journal} {\bibinfo
  {journal} {Nat. Commun.}\ }\textbf {\bibinfo {volume} {7}},\ \bibinfo {pages}
  {10138} (\bibinfo {year} {2016})}\BibitemShut {NoStop}%
\bibitem [{\citenamefont {Rebentrost}\ \emph {et~al.}(2018)\citenamefont
  {Rebentrost}, \citenamefont {Bromley}, \citenamefont {Weedbrook},\ and\
  \citenamefont {Lloyd}}]{q-hopfield}%
  \BibitemOpen
  \bibfield  {author} {\bibinfo {author} {\bibfnamefont {P.}~\bibnamefont
  {Rebentrost}}, \bibinfo {author} {\bibfnamefont {T.~R.}\ \bibnamefont
  {Bromley}}, \bibinfo {author} {\bibfnamefont {C.}~\bibnamefont {Weedbrook}},
  \ and\ \bibinfo {author} {\bibfnamefont {S.}~\bibnamefont {Lloyd}},\
  }\bibfield  {title} {\enquote {\bibinfo {title} {{Quantum Hopfield neural
  network}},}\ }\href {\doibase 10.1103/PhysRevA.98.042308} {\bibfield
  {journal} {\bibinfo  {journal} {Phys. Rev. A}\ }\textbf {\bibinfo {volume}
  {98}},\ \bibinfo {pages} {042308} (\bibinfo {year} {2018})}\BibitemShut
  {NoStop}%
\bibitem [{\citenamefont {{A. Monras, A. Beige, and K.
  Wiesner}}(2010)}]{monras-hmm}%
  \BibitemOpen
  \bibfield  {author} {\bibinfo {author} {\bibnamefont {{A. Monras, A. Beige,
  and K. Wiesner}}},\ }\bibfield  {title} {\enquote {\bibinfo {title} {{Hidden
  Quantum Markov Models and non-adaptive read-out of many-body states}},}\
  }\href {https://arxiv.org/abs/1002.2337} {\bibfield  {journal} {\bibinfo
  {journal} {arXiv:1002.2337}\ } (\bibinfo {year} {2010})}\BibitemShut
  {NoStop}%
\bibitem [{\citenamefont {Schuld}\ \emph
  {et~al.}(2014{\natexlab{a}})\citenamefont {Schuld}, \citenamefont
  {Sinayskiy},\ and\ \citenamefont {Petruccione}}]{schuld-q-walk}%
  \BibitemOpen
  \bibfield  {author} {\bibinfo {author} {\bibfnamefont {M.}~\bibnamefont
  {Schuld}}, \bibinfo {author} {\bibfnamefont {I.}~\bibnamefont {Sinayskiy}}, \
  and\ \bibinfo {author} {\bibfnamefont {F.}~\bibnamefont {Petruccione}},\
  }\bibfield  {title} {\enquote {\bibinfo {title} {{Quantum walks on graphs
  representing the firing patterns of a quantum neural network}},}\ }\href
  {\doibase 10.1103/PhysRevA.89.032333} {\bibfield  {journal} {\bibinfo
  {journal} {Phys. Rev. A}\ }\textbf {\bibinfo {volume} {89}},\ \bibinfo
  {pages} {032333} (\bibinfo {year} {2014}{\natexlab{a}})}\BibitemShut
  {NoStop}%
\bibitem [{\citenamefont {Schuld}\ \emph
  {et~al.}(2014{\natexlab{b}})\citenamefont {Schuld}, \citenamefont
  {Sinayskiy},\ and\ \citenamefont {Petruccione}}]{quest_qnn}%
  \BibitemOpen
  \bibfield  {author} {\bibinfo {author} {\bibfnamefont {M.}~\bibnamefont
  {Schuld}}, \bibinfo {author} {\bibfnamefont {I.}~\bibnamefont {Sinayskiy}}, \
  and\ \bibinfo {author} {\bibfnamefont {F.}~\bibnamefont {Petruccione}},\
  }\bibfield  {title} {\enquote {\bibinfo {title} {{The quest for a Quantum
  Neural Network}},}\ }\href {\doibase 10.1007/s11128-014-0809-8} {\bibfield
  {journal} {\bibinfo  {journal} {{Quantum Inf. Process.}}\ }\textbf {\bibinfo
  {volume} {13}},\ \bibinfo {pages} {2567} (\bibinfo {year}
  {2014}{\natexlab{b}})}\BibitemShut {NoStop}%
\bibitem [{\citenamefont {Rotondo}\ \emph {et~al.}(2018)\citenamefont
  {Rotondo}, \citenamefont {Marcuzzi}, \citenamefont {Garrahan}, \citenamefont
  {Lesanovsky},\ and\ \citenamefont {M{\"u}ller}}]{open-hopfield}%
  \BibitemOpen
  \bibfield  {author} {\bibinfo {author} {\bibfnamefont {P.}~\bibnamefont
  {Rotondo}}, \bibinfo {author} {\bibfnamefont {M.}~\bibnamefont {Marcuzzi}},
  \bibinfo {author} {\bibfnamefont {J.~P.}\ \bibnamefont {Garrahan}}, \bibinfo
  {author} {\bibfnamefont {I.}~\bibnamefont {Lesanovsky}}, \ and\ \bibinfo
  {author} {\bibfnamefont {M.}~\bibnamefont {M{\"u}ller}},\ }\bibfield  {title}
  {\enquote {\bibinfo {title} {{Open quantum generalisation of Hopfield neural
  networks}},}\ }\href {http://stacks.iop.org/1751-8121/51/i=11/a=115301}
  {\bibfield  {journal} {\bibinfo  {journal} {J. Phys. A}\ }\textbf {\bibinfo
  {volume} {51}},\ \bibinfo {pages} {115301} (\bibinfo {year}
  {2018})}\BibitemShut {NoStop}%
\bibitem [{\citenamefont {{J.-ichi Inoue}}(2011)}]{q-hopfield2}%
  \BibitemOpen
  \bibfield  {author} {\bibinfo {author} {\bibnamefont {{J.-ichi Inoue}}},\
  }\bibfield  {title} {\enquote {\bibinfo {title} {{Pattern-recalling processes
  in quantum Hopfield networks far from saturation}},}\ }\href
  {http://stacks.iop.org/1742-6596/297/i=1/a=012012} {\bibfield  {journal}
  {\bibinfo  {journal} {J. of Phys.: Conference Series}\ }\textbf {\bibinfo
  {volume} {297}},\ \bibinfo {pages} {012012} (\bibinfo {year}
  {2011})}\BibitemShut {NoStop}%
\bibitem [{\citenamefont {Kearns}\ and\ \citenamefont
  {Vazirani}(1994)}]{kearns-vazirani}%
  \BibitemOpen
  \bibfield  {author} {\bibinfo {author} {\bibfnamefont {M.~J.}\ \bibnamefont
  {Kearns}}\ and\ \bibinfo {author} {\bibfnamefont {U.~V.}\ \bibnamefont
  {Vazirani}},\ }\href
  {https://mitpress.mit.edu/books/introduction-computational-learning-theory}
  {\emph {\bibinfo {title} {{An Introduction to Computational Learning
  Theory}}}}\ (\bibinfo  {publisher} {MIT Press},\ \bibinfo {year}
  {1994})\BibitemShut {NoStop}%
\bibitem [{\citenamefont {Wang}\ and\ \citenamefont
  {Gertler}(2018)}]{dissi-state-transfer}%
  \BibitemOpen
  \bibfield  {author} {\bibinfo {author} {\bibfnamefont {C.}~\bibnamefont
  {Wang}}\ and\ \bibinfo {author} {\bibfnamefont {J.~M.}\ \bibnamefont
  {Gertler}},\ }\bibfield  {title} {\enquote {\bibinfo {title} {{Directional
  Transfer of Quantum Information by Dissipation Engineering}},}\ }\href
  {https://arxiv.org/abs/1809.03571} {\bibfield  {journal} {\bibinfo  {journal}
  {arXiv:1809.03571}\ } (\bibinfo {year} {2018})}\BibitemShut {NoStop}%
\bibitem [{\citenamefont {Zanardi}\ and\ \citenamefont
  {Rasetti}(1997)}]{Zanardi:97c}%
  \BibitemOpen
  \bibfield  {author} {\bibinfo {author} {\bibfnamefont {P.}~\bibnamefont
  {Zanardi}}\ and\ \bibinfo {author} {\bibfnamefont {M.}~\bibnamefont
  {Rasetti}},\ }\bibfield  {title} {\enquote {\bibinfo {title} {Noiseless
  quantum codes},}\ }\href
  {http://link.aps.org/doi/10.1103/PhysRevLett.79.3306} {\bibfield  {journal}
  {\bibinfo  {journal} {Phys. Rev. Lett.}\ }\textbf {\bibinfo {volume} {79}},\
  \bibinfo {pages} {3306} (\bibinfo {year} {1997})}\BibitemShut {NoStop}%
\bibitem [{\citenamefont {Lidar}\ \emph {et~al.}(1998)\citenamefont {Lidar},
  \citenamefont {Chuang},\ and\ \citenamefont {Whaley}}]{Lidar:1998fk}%
  \BibitemOpen
  \bibfield  {author} {\bibinfo {author} {\bibfnamefont {D.~A.}\ \bibnamefont
  {Lidar}}, \bibinfo {author} {\bibfnamefont {I.~L.}\ \bibnamefont {Chuang}}, \
  and\ \bibinfo {author} {\bibfnamefont {K.~B.}\ \bibnamefont {Whaley}},\
  }\bibfield  {title} {\enquote {\bibinfo {title} {Decoherence-free subspaces
  for quantum computation},}\ }\href
  {http://link.aps.org/doi/10.1103/PhysRevLett.81.2594} {\bibfield  {journal}
  {\bibinfo  {journal} {Phys. Rev. Lett.}\ }\textbf {\bibinfo {volume} {81}},\
  \bibinfo {pages} {2594} (\bibinfo {year} {1998})}\BibitemShut {NoStop}%
\bibitem [{\citenamefont {Valiant}(1984)}]{Valiant:PAC}%
  \BibitemOpen
  \bibfield  {author} {\bibinfo {author} {\bibfnamefont {L.~G.}\ \bibnamefont
  {Valiant}},\ }\bibfield  {title} {\enquote {\bibinfo {title} {{A Theory of
  the Learnable}},}\ }\href {\doibase 10.1145/1968.1972} {\bibfield  {journal}
  {\bibinfo  {journal} {Commun. ACM}\ }\textbf {\bibinfo {volume} {27}},\
  \bibinfo {pages} {1134} (\bibinfo {year} {1984})}\BibitemShut {NoStop}%
\bibitem [{\citenamefont {Bshouty}\ and\ \citenamefont
  {Jackson}(1999)}]{qpac-bshouty-jackson}%
  \BibitemOpen
  \bibfield  {author} {\bibinfo {author} {\bibfnamefont {N.~H.}\ \bibnamefont
  {Bshouty}}\ and\ \bibinfo {author} {\bibfnamefont {J.~C.}\ \bibnamefont
  {Jackson}},\ }\bibfield  {title} {\enquote {\bibinfo {title} {{Learning DNF
  over the Uniform Distribution Using a Quantum Example Oracle}},}\ }\href
  {https://doi.org/10.1137/S0097539795293123} {\bibfield  {journal} {\bibinfo
  {journal} {SIAM J. Comput.}\ }\textbf {\bibinfo {volume} {28}},\ \bibinfo
  {pages} {1136} (\bibinfo {year} {1999})}\BibitemShut {NoStop}%
\bibitem [{\citenamefont {Lindblad}(1976)}]{Lindblad:76}%
  \BibitemOpen
  \bibfield  {author} {\bibinfo {author} {\bibfnamefont {G.}~\bibnamefont
  {Lindblad}},\ }\bibfield  {title} {\enquote {\bibinfo {title} {{On the
  Generators of Quantum Dynamical Semigroups}},}\ }\href {\doibase
  10.1007/BF01608499} {\bibfield  {journal} {\bibinfo  {journal} {Comm. Math.
  Phys.}\ }\textbf {\bibinfo {volume} {48}},\ \bibinfo {pages} {119} (\bibinfo
  {year} {1976})}\BibitemShut {NoStop}%
\bibitem [{\citenamefont {{V. Gorini, A. Kossakowski and E.C.G
  Sudarshan}}(1976)}]{Gorini:76}%
  \BibitemOpen
  \bibfield  {author} {\bibinfo {author} {\bibnamefont {{V. Gorini, A.
  Kossakowski and E.C.G Sudarshan}}},\ }\bibfield  {title} {\enquote {\bibinfo
  {title} {{Completely positive dynamical semigroups of N-level systems}},}\
  }\href {https://doi.org/10.1063/1.522979} {\bibfield  {journal} {\bibinfo
  {journal} {J. Math. Phys.}\ }\textbf {\bibinfo {volume} {17}},\ \bibinfo
  {pages} {821} (\bibinfo {year} {1976})}\BibitemShut {NoStop}%
\bibitem [{Note1()}]{Note1}%
  \BibitemOpen
  \bibinfo {note} {{$\sigma ^- = |0\delimiter "526930B \delimiter "426830A 1|$
  is defined in the eigen-basis of $\sigma ^z = |1\delimiter "526930B
  \delimiter "426830A 1 | - |0\delimiter "526930B \delimiter "426830A
  0|$.}}\BibitemShut {Stop}%
\bibitem [{Note2()}]{Note2}%
  \BibitemOpen
  \bibinfo {note} {{We wish to find $t$ such that $(1-e^{-\gamma t})^N >
  1-\epsilon $. Noting that $e^{-\epsilon } > 1-\epsilon $, we can instead
  bound $(1-e^{-\gamma t})^N > e^{-\epsilon }$. Using the identity $-\protect
  \qopname \relax o{log}(1-x) < \protect \frac {x}{1-x}$ with $x=e^{-\gamma t}$
  gives $\protect \frac {Ne^{-\gamma t}}{1-e^{-\gamma t}} < \epsilon $ which
  can be rearranged to give the result in the main text, $t > O(\protect \frac
  {1}{\gamma }\protect \qopname \relax o{log}\protect \frac {N}{\epsilon
  })$.}}\BibitemShut {Stop}%
\bibitem [{Note3()}]{Note3}%
  \BibitemOpen
  \bibinfo {note} {{Here $\delimiter "026B30D \cdot \delimiter "026B30D _1$ is
  the trace norm: $\delimiter "026B30D X\delimiter "026B30D _1 = \protect
  \mathrm {Tr}|X|$ with $|X|=\protect \sqrt {X^\dagger X}$.}}\BibitemShut
  {Stop}%
\bibitem [{\citenamefont {Zadah}(1965)}]{fuzzy-logic}%
  \BibitemOpen
  \bibfield  {author} {\bibinfo {author} {\bibfnamefont {L.~A.}\ \bibnamefont
  {Zadah}},\ }\bibfield  {title} {\enquote {\bibinfo {title} {Fuzzy sets},}\
  }\href {https://doi.org/10.1016/S0019-9958(65)90241-X} {\bibfield  {journal}
  {\bibinfo  {journal} {Information and Control}\ }\textbf {\bibinfo {volume}
  {8}},\ \bibinfo {pages} {338} (\bibinfo {year} {1965})}\BibitemShut {NoStop}%
\bibitem [{\citenamefont {Russel}\ and\ \citenamefont
  {Norvig}(2010)}]{aima-book}%
  \BibitemOpen
  \bibfield  {author} {\bibinfo {author} {\bibfnamefont {S.}~\bibnamefont
  {Russel}}\ and\ \bibinfo {author} {\bibfnamefont {P.}~\bibnamefont
  {Norvig}},\ }\href
  {https://www.pearson.com/us/higher-education/program/Russell-Artificial-Intelligence-A-Modern-Approach-3rd-Edition/PGM156683.html?tab=overview}
  {\emph {\bibinfo {title} {Artificial Intelligence: A Modern Approach}}},\
  \bibinfo {edition} {3rd}\ ed.\ (\bibinfo  {publisher} {Pearson},\ \bibinfo
  {year} {2010})\BibitemShut {NoStop}%
\bibitem [{\citenamefont {A{\"i}meur}\ \emph {et~al.}(2006)\citenamefont
  {A{\"i}meur}, \citenamefont {Brassard},\ and\ \citenamefont
  {Gambs}}]{ML-Q_world}%
  \BibitemOpen
  \bibfield  {author} {\bibinfo {author} {\bibfnamefont {E.}~\bibnamefont
  {A{\"i}meur}}, \bibinfo {author} {\bibfnamefont {G.}~\bibnamefont
  {Brassard}}, \ and\ \bibinfo {author} {\bibfnamefont {S.}~\bibnamefont
  {Gambs}},\ }\bibfield  {title} {\enquote {\bibinfo {title} {{Machine Learning
  in a Quantum World}},}\ }in\ \href
  {https://link.springer.com/chapter/10.1007/11766247_37} {\emph {\bibinfo
  {booktitle} {Advances in Artificial Intelligence}}}\ (\bibinfo  {publisher}
  {Springer Berlin Heidelberg},\ \bibinfo {address} {Berlin, Heidelberg},\
  \bibinfo {year} {2006})\ pp.\ \bibinfo {pages} {431--442}\BibitemShut
  {NoStop}%
\bibitem [{\citenamefont {A{\"i}meur}\ \emph {et~al.}(2013)\citenamefont
  {A{\"i}meur}, \citenamefont {Brassard},\ and\ \citenamefont
  {Gambs}}]{unsupervised-speedup}%
  \BibitemOpen
  \bibfield  {author} {\bibinfo {author} {\bibfnamefont {E.}~\bibnamefont
  {A{\"i}meur}}, \bibinfo {author} {\bibfnamefont {G.}~\bibnamefont
  {Brassard}}, \ and\ \bibinfo {author} {\bibfnamefont {S.}~\bibnamefont
  {Gambs}},\ }\bibfield  {title} {\enquote {\bibinfo {title} {Quantum speed-up
  for unsupervised learning},}\ }\href {\doibase 10.1007/s10994-012-5316-5}
  {\bibfield  {journal} {\bibinfo  {journal} {Machine Learning}\ }\textbf
  {\bibinfo {volume} {90}},\ \bibinfo {pages} {261} (\bibinfo {year}
  {2013})}\BibitemShut {NoStop}%
\bibitem [{\citenamefont {Simon}(1994)}]{simon:94}%
  \BibitemOpen
  \bibfield  {author} {\bibinfo {author} {\bibfnamefont {D.~R.}\ \bibnamefont
  {Simon}},\ }\bibfield  {title} {\enquote {\bibinfo {title} {{On the power of
  quantum computation}},}\ }in\ \href
  {https://doi.org/10.1109/SFCS.1994.365701} {\emph {\bibinfo {booktitle}
  {Proc. 35th Annual Symposium on Foundations of Computer Science}}}\ (\bibinfo
  {year} {1994})\ pp.\ \bibinfo {pages} {116--123}\BibitemShut {NoStop}%
\bibitem [{\citenamefont {Servedio}(2001)}]{servedio}%
  \BibitemOpen
  \bibfield  {author} {\bibinfo {author} {\bibfnamefont {R.~A.}\ \bibnamefont
  {Servedio}},\ }\bibfield  {title} {\enquote {\bibinfo {title} {{Separating
  Quantum and Classical Learning}},}\ }in\ \href
  {https://link.springer.com/chapter/10.1007/3-540-48224-5_86} {\emph {\bibinfo
  {booktitle} {Automata, Languages and Programming}}}\ (\bibinfo  {publisher}
  {Springer Berlin Heidelberg},\ \bibinfo {address} {Berlin, Heidelberg},\
  \bibinfo {year} {2001})\ pp.\ \bibinfo {pages} {1065--1080}\BibitemShut
  {NoStop}%
\bibitem [{Note4()}]{Note4}%
  \BibitemOpen
  \bibinfo {note} {The same arguments as in Appendix \ref {app:bound} apply to
  Prop.~\ref {prop3} regarding the evolution time as a function of error
  $\epsilon $, which is therefore still $t\sim O(\protect \frac {n}{\gamma
  }+\protect \frac {1}{\gamma }\protect \qopname \relax o{log}\protect \frac
  {N}{\epsilon })$.}\BibitemShut {Stop}%
\bibitem [{Note5()}]{Note5}%
  \BibitemOpen
  \bibinfo {note} {As mentioned in point 4 following Prop.~\ref {prop1}, our
  framework can equivalently be used to evaluate both CNF and DNF
  formulas.}\BibitemShut {Stop}%
\bibitem [{\citenamefont {Zhang}\ \emph {et~al.}(2017)\citenamefont {Zhang},
  \citenamefont {Wagenbreth}, \citenamefont {Martin-Mayor},\ and\ \citenamefont
  {Hen}}]{fairInUnfair}%
  \BibitemOpen
  \bibfield  {author} {\bibinfo {author} {\bibfnamefont {B.~H.}\ \bibnamefont
  {Zhang}}, \bibinfo {author} {\bibfnamefont {G.}~\bibnamefont {Wagenbreth}},
  \bibinfo {author} {\bibfnamefont {V.}~\bibnamefont {Martin-Mayor}}, \ and\
  \bibinfo {author} {\bibfnamefont {I.}~\bibnamefont {Hen}},\ }\bibfield
  {title} {\enquote {\bibinfo {title} {{Advantages of Unfair Quantum
  Ground-State Sampling}},}\ }\href {\doibase 10.1038/s41598-017-01096-6}
  {\bibfield  {journal} {\bibinfo  {journal} {Sci. Rep.}\ }\textbf {\bibinfo
  {volume} {7}},\ \bibinfo {pages} {1044} (\bibinfo {year} {2017})}\BibitemShut
  {NoStop}%
\bibitem [{\citenamefont {Arunachalam}\ and\ \citenamefont
  {de~Wolf}(2017)}]{qpac-optimal}%
  \BibitemOpen
  \bibfield  {author} {\bibinfo {author} {\bibfnamefont {S.}~\bibnamefont
  {Arunachalam}}\ and\ \bibinfo {author} {\bibfnamefont {R.}~\bibnamefont
  {de~Wolf}},\ }\bibfield  {title} {\enquote {\bibinfo {title} {{Optimal
  Quantum Sample Complexity of Learning Algorithms}},}\ }in\ \href {\doibase
  10.4230/LIPIcs.CCC.2017.25} {\emph {\bibinfo {booktitle} {Proc. 32nd
  Computational Complexity Conference}}},\ \bibinfo {series and number} {CCC
  '17}\ (\bibinfo  {publisher} {Schloss Dagstuhl--Leibniz-Zentrum fuer
  Informatik},\ \bibinfo {address} {Germany},\ \bibinfo {year} {2017})\ pp.\
  \bibinfo {pages} {25:1--25:31}\BibitemShut {NoStop}%
\end{thebibliography}%

\appendix

\section{Proof of Eq.~(\ref{eq:inf-time}) \label{app:bound}}
Let $\rho_t^{(xy)} = \mathcal{E}_t(|x\rangle \langle y|\otimes |1\rangle \langle 1|^{\otimes N})$. By the triangle inequality and linearity of $\mathcal{E}_t$, we have $\|\rho_t - \rho_\infty\|_1 \le 2^{2n}\| \rho_t^{(xy)} - \rho_{\infty}^{(xy)}\|_1$. We set $\| \rho_t^{(xy)} - \rho_{\infty}^{(xy)}\|_1 < \epsilon/2^{2n}$ and directly compute this norm. There are two general cases i) $\bar{C}(x)=\bar{C}(y)$, ii) $\bar{C}(x) \neq \bar{C}(y)$.

In i), let us first consider the case where $\bar{C}_i(x) = 0 = \bar{C}_i(y), \forall i$. Then $\rho_\infty^{(xy)} = |x\rangle \langle y| \otimes |\bar{0}\rangle \langle \bar{0}|$ where $|\bar{0}\rangle \equiv |0\rangle^{\otimes N}$. Then, using Eq.~(\ref{eq:proof2}),
\begin{equation}
\begin{split}
    (\rho_t^{(xy)} - \rho_\infty^{(xy)})^\dagger (\rho_t^{(xy)} - \rho_\infty^{(xy)}) \\
    =\sum_{\bar{C} \in \{0,1\}^N} \alpha_{\bar{C}}^2 |y\rangle \langle y| \otimes |\bar{C}\rangle \langle \bar{C}|
    \end{split}
    \label{eq:norm1}
\end{equation}

with 
\begin{equation}
    \begin{split}
    & \alpha_{\bar{C} \neq \bar{0}} = e^{-N_{\bar{C}}^{(0)} \gamma t}(1-e^{-\gamma t})^{N-N_{\bar{C}}^{(0)} }\\
        & \alpha_{\bar{0}} = \sqrt{(1-e^{-\gamma t})^{2N} - 2(1-e^{-\gamma t})^N + 1} \\
        & = 1-(1-e^{-\gamma t})^N
    \end{split}
\end{equation}
where $N_{\bar{C}}^{(0)} $ is the number of $0$'s in the bit-string $\bar{C}$. The form of $a_{\bar{0}}$ comes from the cross terms in Eq.~(\ref{eq:norm1}).

Taking the matrix square root of Eq.~(\ref{eq:norm1}) (which is diagonal), followed by the trace, gives
\begin{equation}
\begin{split}
    \| \rho_t^{(xy)} - \rho_{\infty}^{(xy)}\|_1 =
    1-(1-e^{-\gamma t})^N +\\ \sum_{k=1}^N {N \choose k} e^{-k \gamma t}(1-e^{-\gamma t})^{N-k} 
    \end{split}
\end{equation}
which can be rearranged (by the Binomial theorem) to
\begin{equation}
    \| \rho_t^{(xy)} - \rho_{\infty}^{(xy)}\|_1 = 2(1 -  (1-e^{-\gamma t})^N).
    \label{eq:bound-simplified}
\end{equation}

Similarly as in the proof of Prop.~\ref{prop2}, to guarantee this is less than $\epsilon/2^{2n}$ requires $t > O(\frac{n}{\gamma}+\frac{1}{\gamma}\log \frac{N}{\epsilon})$ as required.

Note, if instead $\bar{C}(x)=\bar{C}(y) \neq \bar{0}$, the only difference is that in Eq.~(\ref{eq:bound-simplified}), $N$ is replaced by $N-H$ where $H$ is the Hamming weight of $\bar{C}(x)$ and $\bar{C}(y)$, and therefore the final conclusion remains unchanged.

For case ii), $\rho_{\infty}^{(xy)}=0$, and if there are $m$ clauses which differ between $x$ and $y$ (and all others evaluate to 0), one instead bounds
\begin{equation}
    e^{-\gamma t m/2}\left[\sum_{k=0}^{N-m}{N-m \choose k} e^{-\gamma k t} (1-e^{-\gamma t})^{N-m-k}\right] < \epsilon/2^{2n},
    \label{eq:bound-differ}
\end{equation}
which by  the Binomial theorem gives $t > O(\frac{1}{\gamma}\log \frac{2^n}{\epsilon})$.
As before, in the case where the are clauses $C_i(x)=C_i(y)=1$, one instead replaces $N$ in Eq.~(\ref{eq:bound-differ}) with $N-H$, where $H$ is the number of such clauses. The result is unchanged.

Overall then, to achieve $\|\rho_t - \rho_\infty \|_1 < \epsilon$,
one requires evolution time $t>O(\frac{n}{\gamma}+\frac{1}{\gamma}\log \frac{N}{\epsilon})$, i.e. Eq.~(\ref{eq:inf-time}).

\end{document}